\documentclass[12pt]{article}
\usepackage{subfigure}
\usepackage{graphicx}
\usepackage[greek,english]{babel}
\usepackage[utf8x]{inputenc}
\usepackage{amssymb,amsmath,amsthm,mathrsfs,bm}
\usepackage[implicit=false]{hyperref}
\usepackage{cite}
\newtheorem{theorem}{Theorem}
\numberwithin{equation}{section}
\newcommand{\newc}{\newcommand}
\newc{\ra}{\rightarrow}
\newc{\lra}{\leftrightarrow}
\newc{\be}{\begin{equation}}
\newc{\ee}{\end{equation}}
\newc{\bg}{\begin{gathered}}
\newc{\eg}{\end{gathered}}
\newc{\bs}{\begin{split}}
\newc{\es}{\end{split}}
\newc{\ba}{\begin{eqnarray}}
\newc{\ea}{\end{eqnarray}}
\newc{\ov}{\overline}
\newc{\pa}{\partial}
\newc{\D}{\Delta}

\begin{document}
\begin{titlepage}
\begin{center}
\begin{Large}
{\bf

 Arnol'd cat map lattices
}

\end{Large}

\hskip1.0truecm

Minos Axenides$^{(a)}$\footnote{E-Mail: axenides@inp.demokritos.gr}, Emmanuel Floratos$^{(a,b)}$\footnote{E-Mail: mflorato@phys.uoa.gr} and Stam Nicolis$^{(c)}$\footnote{E-Mail: stam.nicolis@lmpt.univ-tours.fr, stam.nicolis@idpoisson.fr}

\hskip1.0truecm
{\sl 
${ }^{(a)}$Institute for Nuclear and Particle Physics, NCSR ``Demokritos''\\Aghia Paraskevi 15310, Greece\\
${ }^{(b)}$Physics Department, University of Athens\\ Athens, 15771 Greece\\
${ }^{(c)}$Institut Denis Poisson, Université de Tours, Université d'Orléans, CNRS\\
Parc Grandmont, 37200 Tours, France

}
\end{center}
\begin{abstract}
We construct Arnol'd cat map lattice field theories in phase space and configuration space. In phase space we impose that the  evolution operator of the linearly coupled maps be an element of the symplectic group, in direct generalization of the case of one map. To this end we exploit the correspondence between the cat map and the Fibonacci sequence. 
The chaotic properties of these systems can be also  understood from the equations of motion in configuration space. These describe inverted harmonic oscillators, where the  runaway behavior of the potential compete with the toroidal compactification of the phase space.
We   highlight  the  spatio-temporal chaotic properties of these systems  using standard  benchmarks for probing deterministic chaos  of dynamical systems, namely  the complete dense  set of  unstable  periodic orbits, which, for  long periods,  lead to ergodicity and mixing. The spectrum of the periods exhibits a  strong dependence on  the strength and the range of the interaction. 
   
 \end{abstract}
\end{titlepage}
\newpage
\section{Introduction and overview}\label{intro}
Low dimensional systems  with few degrees of freedom have provided a fertile ground for the development of the  concepts and methods  of deterministic chaos with their characteristic  disordered behaviour at the classical and quantum level~\cite{arnold2007mathematical,cvitanovic2005chaos,zaslavsky2008,gutzwiller2013chaos}. 
While the original focus of interest centered around  dissipative and Hamiltonian low dimensional systems, progress  was quickly followed by efforts to understand the complex dynamics of  high dimensional systems consisting  of many coupled chaotic degrees of freedom~\cite{crutchfield1987phenomenology,kaneko2001complex,vulpiani2009chaos,badiicomplexity}.
They are spatially extended systems, which can be  driven away from equilibrium and   exhibit spatio-temporal chaos (STC). They give rise to  diverse pattern formation ~\cite{cross1994spatiotemporal, cross1993pattern}  as the result of their highly complex dynamical behavior. 
STC models possess either continuous or discrete time dynamics (maps). The spatial degrees of freedom are either discrete or continuous giving rise, respectively, to lattice dynamics or an effective hydrodynamic description in terms of continuous fields.  Such systems described are described by their equations of motion:  Partial differential equations, systems of coupled ordinary differential equations;  or as  coupled map lattices (CMLs) with continuous state spaces. Another way is  using    cellular automata with discrete state spaces ~\cite{Kaneko:2014}.

Indicators for chaos for spatiotemporal  systems have been proposed--namely the finite amplitude Lyapunov exponents, covariant  Lyapunov vector exponents ~\cite{pikovsky2016lyapunov}, as well as benchmark dynamical entropies,  like the  Kolmogorov-Sinai entropy~\cite{Sinai1994,cornfeld2012ergodic}.

While the approach to the problem of describing  spatio-temporal chaos  in coupled map lattices  is mostly numerical and a comprehensive  understanding from the analytical side is, still, lacking,  there has been some activity recently in an effort to acquire an analytical understanding of    the dynamics of {\em linear}  CMLs ~\cite{cvitanovic2000chaotic,gutkin2016,gutkin:2019sca,cvitanovic2020spatiotemporal}. 
One  aim of this program of research is to define chaotic field theories made up  of  chaotic oscillator constituents,  in an effort to provide a local description of  some of the coherent structures that emerge from  the dynamics of continuous fluid systems in the régime of weak turbulent flows~\cite{holmes2012turbulence}.
However it isn't clear, whether the complexity of these structures is due to the known complex behavior of their constituents, the result of the way they are coupled, or both. The reason is that the typical way for establishing such a relation, namely the study of symmetries, has proven to be very difficult to follow for these systems. 

There are, however, cases where this approach is possible. Our present work focuses  on the systematic construction of a special class of CMLs, the lattice field theories of Arnol'd cat maps in various dimensions, taking into account their symmetries in phase space--namely covariance under symplectic transformations--and how these are related to the symmetries in configuration space. 

In the present paper we  provide  the classical framework for describing the dynamics of chaotic oscillators  via the dynamics of    $n,$ linearly coupled, Arnol'd cat maps (CACML), subject to periodic boundary conditions. Their  phase space is the torus  $\mathbb{T}^{2n}[\mathbb{Z}].$ and their dynamics is represented by elements  ${\sf M}$ of the symplectic group, $\mathrm{Sp}_{2n}[\mathbb{Z}].$ This is the generalization for $n$ cat maps, of the symmetry properties of one cat map. 
Within this framework  we can vary the dimensionality of the lattice, the number of oscillators, the strength  of their interactions as well as the range of the non-locality thereof.

We can therefore study in detail the classical chaotic properties of this system, since it's possible to obtain explicit expressions, that can be reliably evaluated. We focus, as an example,  on  the classical spatio-temporal chaotic properties of CACML's in one dimension, using a  benchmark of chaos of any chaotic system, namely, the set of all of its unstable periodic orbits~\cite{auerbach1987exploring,cvitanovic2005chaos}.

The periodic orbits of the CACMLs  are classified by  initial conditions,  which have  rational  coordinates in the toroidal phase space, with  common denominator  $N.$  Upon varying $N=3,5,7,\ldots,$ over the primes and, more, generally, the odd integers (even integers have subtle issues, particularly in the quantum case~\cite{Floratos:2005yj,kaiblinger2009metaplectic},  so require special study) we obtain all the periodic orbits (which  are all unstable)~\cite{Sinai1994}. 
For large $N$  and for fixed size of the toroidal phase space  we can approach a scaling limit.  For  long periods the periodic trajectories  lead to ergodicity and strong  mixing~\cite{Sinai1994,cornfeld2012ergodic,sep-ergodic-hierarchy}. 

This limit can be subtle, already for one map~\cite{Axenides:2019lea}. 

In the case of  translational invariant couplings we find explicitly: a) all the periodic orbits, b) the Lyapunov spectra and c)the Kolmogorov--Sinai entropy of the CACMLs as a function of the strength and the  range of interactions. Armed with these analytic results we find that the maximum Lyapunov exponent of these systems is an increasing  function of both the coupling constant and of the range  of the interaction.

We provide also a method for determining the periods of the orbits based on the properties of  matrix Fibonacci polynomials. These periods are random functions of $N$ and they have stronger dependence on the coupling and the range of the  interaction than in   the non interacting case, i.e. for the single cat map. We present  several numerical examples in support of this observation. The dependence of the periods on $N$  provides information about the quantum spectra of these systems, which deserve a study in their own right. 

For the case of the single cat map, which corresponds to $n=1,$ 
a detailed study of the periods, of their relation to the energy spectrum and its asymptotic properties  for the quantum system, can be found   in ref.~\cite{percival1987linear,percival1987arithmetical,bird1988periodic,Keating1991}.

Now we would like to discuss our particular motivation for this study. 

This  derives from the realization that the physics of quantum black holes is a prime example of a chaotic many-body system, when the microstates can be resolved. Therefore it has become of topical interest to construct models for both the probes and for the near horizon geometry, that is defined by the microstates.  This is why a consistent description of chaotic field theories has become a fascinating bridge that establishes novel relations between the subjects of interest to  the high energy community and the community of classical and quantum chaos. This can be summarized as follows: 
 
Black Holes (BH) are at present understood to be physical systems of finite  entropy which, for an observer at infinity,  is described by  the dynamics of the microstates of the black hole, that live in the near horizon geometry. The chaotic  dynamics of these microstates has new features, such as fast scrambling and non-locality.  Specifically 
it has been conjectured  that black holes are the fastest information scramblers in Nature ~\cite{Hayden:2007cs,Sekino:2008he,Shenker:2013pqa,Maldacena:2015waa,Bousso:2022ntt},  that exhibit  unitary quantum evolution. 
  This, in turn, motivated the search for  models that can capture these features.  One class of  such models builds upon  the relation between  the near horizon shock wave geometries and of the so-called gravitational memory effects.
In these models it seems, indeed,  possible--in principle--that the near horizon region of a black hole  could form a chaotic memory, i.e. a basin of purely geometrical data  of all of its past and recent  history, through the 't Hooft mechanism of permanent space-time displacements caused by high energy scattering events of infalling  wave packets~\cite{Hooft:2015jea,Barrabes:2000fr,Papadodimas:2012aq,Banerjee:2016mhh,Stanford:2014jda,Shenker:2013yza,Mezei:2016wfz,Polchinski:2015cea}.
In the language of refs.~\cite{Strominger:2014pwa,Ellis:2016atb,Hawking:2016msc} such data can be identified with  the soft  hair of the BH, whose origin is the infinite number of conservation laws, described by the BMS group.
Proposals for a chaotic dynamics, within a discretized  spacetime,  for the microscopic degrees of freedom of the stretched horizon have been discussed for quite some time in the literature~\cite{Hooft:2016pmw,Brown:2016wib,Banks:1997hz,Iizuka:2008eb,Avery:2011nb,Magan:2016ojb}.

Our contribution to this quest started with  the study of   single particle probes, sent by observers at infinity, in order to learn about  the near horizon AdS$_2$ geometries of  black holes, taken as  discrete and nonlocal dynamical  systems ~\cite{Axenides:2013iwa,Axenides:2015aha,Axenides:2016nmf}. 

More specifically we have shown how the so-called Arnol'd cat maps, acting in a AdS$_2$ discrete near horizon geometry, can capture the properties of its single particle probes.
 We constructed explicitly  an exact discrete version of AdS$_2$/CFT$_1$ correspondence with chaotic and mixing dynamics for Gaussian single-particle wave packets, that is shown to provide an example of the so--called ``Eigenstate Thermalization Hypothesis''~\cite{SrednickiETH}.  Finally, we have demonstrated that the model for  their discrete and chaotic~\cite{Arnold11}, near horizon geometry admits a continuum limit ~\cite{Axenides:2019lea}, where the smooth classical geometry is recovered.

The long term objective of our recent work is to provide models of non-local chaotic quantum  dynamics of the tuneable rate of mixing(and its quantum avatar, scrambling) for  the  degrees of the horizon itself  by $n-$particle systems. Our conjecture is that this can  be achieved through the construction  of the quantum  CMLs  of Arnol'd cat maps~\cite{mantica2019many}.

Therefore, while our previous work focused on the properties of single particle probes of the near horizon geometry,  in the present work,  we construct many-body systems, that possess the necessary features  expected of the interacting black hole microstates themselves--namely, non-locality, chaos and strong mixing (scrambling). Therefore these many-body systems can be considered as effective models of the dynamics of the near horizon geometry itself.  

Below we present the plan and summarize the results of the paper.

In section~\ref{SymplAutom} we present the general setting of the dynamics of  systems of $n$ particles with evolution maps that are  integral toral automorphisms of the  $2 n$ dimensional  phase space,$\mathbb{T}^{2n}$ i.e elements of the symplectic group  Sp$_{2n}[\mathbb{Z}],$ acting on points of  the torus $\mathbb{T}^{2n}$ of radius $R\equiv 1\,\mathrm{mod}\,1.$ 

The completely chaotic and mixing dynamics is described by the maximally hyperbolic elements of this group i.e. whose  eigenvalues are pairs of positive real numbers,$(\lambda>1,1/\lambda<1)$, thus decomposing the phase space into symplectic planes with hyperbolic motion~\cite{Sinai1994,zaslavsky2008,mcduff2017introduction}. 

In section~\ref{symplfibseq}  we discuss  how to obtain elements of Sp$_{2n}[\mathbb{Z}],$which describe $n$ coupled Arnold cat maps that are maximally hyperbolic.
 Starting from the most  general way for linearly coupling   Fibonacci integer sequences, we  construct  a family  of coupled Arnold cat maps lattices (CACML), in any dimension  $d=1,2,..$  imposing  symplectic interactions of  tuneable non-locality and we show that they are all  maximally  hyperbolic toral automorphisms.

In section~\ref{eom} for  the case of translation invariant couplings, in $d=1,$ we determine explicitly all the orbits of the CACML (periodic and non periodic).

In section~\ref{Lyap}  we  discuss further measures of STC, namely the Lyapunov spectra and the  Kolmogorov-Sinai entropy and find that the latter  scales as the volume of the system.  This property is a significant check of the consistency of our calculations and shows that the CACML defined in this way does have sensible thermodynamics.
We observe also that the Kolmogorov-Sinai entropy is a good  proxy  for the mixing time(scrambling)  of the dynamical system: The bigger the K-S entropy the faster the mixing time, so tuning the K-S entropy with the parameters of the system we can tune its mixing time. 
The corresponding property for the quantum system pertains to the entanglement entropy of subsystems and its time evolution  (cf. also~\cite{Bianchi:2017kgb,de2022linear}). 

We discuss  the dependence of the Lyapunov spectra on the dimensionality,the size of the system, $n,$ the strength $G,$ and the range, $l,$ of the interactions.

In section~\ref{permodN}  we   determine  the periods of the  periodic orbits.

To do that we discretize the toroidal phase space by considering all the initial conditions, which are  rational numbers with a common demoninator $N.$
This new phase space we call  $\mathbb{T}^{2n}[N].$
In this--discrete--phase space the toral automorphisms are elements of the group  Sp$_{2n}[\mathbb{Z}_N].$
The  set of all periodic orbits of the corresponding   dynamical systems on the continuous phase space $\mathbb{T}^{2n}[\mathbb{R}],$ are given by the set of all different orbits of the CACML in    $\mathbb{T}^{2n}[N]$   mod  $N,$ by considering all possible values of $N.$

The spectra of the  periods $T[N]$ of the CACML  are the lengths of its orbits and they are random functions of $N.$ They are determined by properties of the  matrix Fibonacci polynomials mod $N.$

We study  numerically  the spectrum  of the periods, for fixed values of the number $n$  of coupled Arnol'd cat maps, the modular integer $N,$ the strength and the range of the interactions.
We observe, as might be  expected, a  random and stronger dependence on $N$  for larger   values of $n,$as well for increasing values of the strength and the range of interactions.

In section~\ref{concl} we present our  conclusions, possible applications, as well as  open problems.

In appendix~\ref{appkfibseq} we review, for completeness, some useful properties of the Fibonacci polynomials and their matrix generalizations and in, appendix~\ref{conslawsmodN} we determine and discuss possible conserved quantities  of the ACML systems,  which can be expressed as  quadratic functions of the position and momenta of the system.The corresponding  conservation laws restrict the volume of the toroidal phase space available to the trajectories of the system  and  lead to the vanishing of some of the Lyapunov exponents of the system (namely through eigenvectors of the evolution operator with eigenvalue one).

\newpage
\section{Dynamics of  symplectic automorphisms of the toroidal phase space $\mathbb{T}^{2n}$ }\label{SymplAutom}

In this section we provide a short review of the mathematical tools, that are necessary for studying the dynamics of $n$ particles, in a  toroidal phase space, $\mathbb{T}^{2n}=\mathbb{R}^{2n}/(\mathbb{Z}^n\times\mathbb{Z}^n)$. The discrete time evolution  will be described by discrete time  maps, ${\sf M},$ that are elements  of the symplectic group, $\mathrm{Sp}_{2n}[\mathbb{Z}]$, that act on the toroidal phase space as
\begin{equation}
\label{TPSevol}
\bm{x}_{m+1}\equiv\bm{x}_m{\sf M}\,\mathrm{mod}\,1
\end{equation}
at the $m$--th time step ($m=0,1,2,\ldots$), 
along with the initial condition $\bm{x}_{m=0}=\bm{x}_0\in\mathbb{T}^{2n}.$ We have taken the length of the sides  of the torus equal to 1.

In this notation, $\bm{x}_m=(\bm{q}_m,\bm{p}_m)$, where $\bm{q}_m$ and $\bm{p}_m$ are the positions and the momenta of the $n$ particles at time step $m$. 

By definition any element ${\sf M}\in\mathrm{Sp}_{2n}[\mathbb{Z}]$ preserves the (symplectic) inner product, $\langle\bm{x}',\bm{x}\rangle$  of any two vectors $\bm{x}$ and $\bm{x}'$,
\begin{equation}
\label{symplinprod}
\langle \bm{x}',\bm{x}\rangle=\sum_{I=1}^n\,\left( q_Ip'_I-q'_Ip_I\right)
\end{equation}
This inner product can be rewritten as
\begin{equation}
\label{innerprod}
\left\langle \bm{x}', \bm{x}\right\rangle={{\bm x}'}^\mathrm{T}{\sf J}\bm{x}
\end{equation}
where ${\sf J}$ is the symplectic matrix
\begin{equation}
\label{Jform}
{\sf J} = \left(
\begin{array}{cc} 0 & -I_{n\times n}\\ I_{n\times n} & 0\end{array}
\right) 
\end{equation}
We can decompose ${\sf M}$ in blocks of $n\times n$ (integer) matrices
\begin{equation}
\label{Mblocks}
{\sf M}=\left( 
\begin{array}{cc} 
{\sf A} & {\sf B} \\ {\sf C} & {\sf D}
\end{array}
\right)
\end{equation}
The invariance of ${\sf J}$ under the action of any element ${\sf M}\in\mathrm{Sp}_{2n}[\mathbb{Z}]$, 
\begin{equation}
\label{symplecticaction}
{\sf J} = {\sf M}^\mathrm{T}{\sf J}{\sf M}
\end{equation} 
implies the constraints
\begin{equation}
\label{symplconstraints}
\begin{array}{l}
{\sf A}^\mathrm{T}{\sf D}-{\sf C}^\mathrm{T}{\sf B} = I_{n\times n}\\
{\sf A}^\mathrm{T}{\sf C}={\sf C}^\mathrm{T}{\sf A}\\
{\sf B}^\mathrm{T}{\sf D}={\sf D}^\mathrm{T}{\sf B}
\end{array}
\end{equation}
The second and third constraints express that ${\sf A}^\mathrm{T}{\sf C}$ and ${\sf B}^\mathrm{T}{\sf D}$ are symmetric, integer--valued, matrices. 

To visualize the motion, in the toroidal phase space, $\mathbb{T}^{2n},$ under the action of ${\sf M},$ it is useful to decompose it into ``simpler'' actions. For it is possible to show that any matrix ${\sf M}\in \mathrm{Sp}_{2n}[\mathbb{Z}]$ can be decomposed into the product of three, simpler, symplectic matrices, that generate the symplectic group, namely
\begin{equation}
\label{sympldecomp}
{\sf M}=\left(\begin{array}{cc} I_{n\times n} & 0 \\ {\sf S}_\mathrm{L} & I_{n\times n}\end{array}\right)
\left(\begin{array}{cc} {\sf U}^\mathrm{T} & 0 \\ 0 & {\sf U}^{-1}\end{array}\right)
\left(\begin{array}{cc} I_{n\times n} & {\sf S}_\mathrm{R} \\ 0 & I_{n\times n}\end{array}\right)
\end{equation}	
where ${\sf S}_\mathrm{R,L}$ are  integer symmetric matrices and ${\sf U}$ an  invertible, integer matrix, which can be determined, given the matrices ${\sf A,B,C,D}$, as follows:
\begin{equation}
\label{ABCDrels}
\begin{array}{l}
{\sf U}={\sf A}^\mathrm{T}\\
{\sf S}_\mathrm{L} = {\sf C}{\sf A}^{-1}\\
{\sf S}_\mathrm{R}={\sf A}^{-1}{\sf B}\\
\end{array}
\end{equation}
These relations hold, iff ${\sf A}$ is invertible, with integer entries. If this isn't the case, it is possible to redefine ${\sf M}$ so that ${\sf A}$ does have the desired properties. 

This decomposition will be useful, also, for the study of the quantum mechanics of this system, since the properties of the unitary evolution operator of the latter are, indeed, those of the metaplectic  representation of the symplectic group~\cite{Athanasiu:1995ni,Axenides:2013iwa}. 

An explicit example is that of  the single Arnol'd cat map ($n=1$), which can be written as 
\begin{equation}
\label{ACMsympldecomp}
{\sf M}=\left(\begin{array}{cc} 1 & 1\\1 & 2\end{array}\right)=\left(\begin{array}{cc} 1&0 \\1 & 1\end{array}\right)1_{2\times 2}\left(\begin{array}{cc}1 & 1 \\0 & 1\end{array}\right)
\end{equation}
and the exercise we shall solve in the following sections is  to find the generalization for $n$ such maps. 

Let us now consider the  initial condition, $\bm{x}_0=(k_1/N,k_2/N,\ldots,k_n/N,l_1/N,l_2/N,\ldots,l_n/N)$, with $0\leq k_I \leq N, 0\leq l_I \leq N$ and $k_I, l_I, N$ integers. 
Since the length of each side of the hypercube that defines $\mathbb{T}^{2n}$ is taken equal to 1, the evolution equation~(\ref{TPSevol}) can be rewritten as
\begin{equation}
\label{TPSevol1}
\left(\bm{k},\bm{l}\right)_{m+1}=\left(\bm{k},\bm{l}\right)_m{\sf M}\,\mathrm{mod}\,N
\end{equation} 

The set of vectors $(\bm{k},\bm{l})\,\mathrm{mod}\,N,$ defines the lattice $\mathbb{Z}_N^n\times\mathbb{Z}_N^n$ which can be identified with $\mathbb{T}^{2n}[N].$  
Since the number of rational points of $\mathbb{T}^{2n}$, with fixed denominator $N$, is equal to $N^{2n}$, the matrix ${\sf M}\,\mathrm{mod}\,N$, which belongs to the finite group Sp$_{2n}[\mathbb{Z}_N]$, has a finite period, $T[N]$, where $T[N]$ is the smallest integer such that ${\sf M}^{T[N]}\equiv I_{n\times n}\,\mathrm{mod}\,N$. 

 This period does not depend, generically, on the initial condition and, thus, defines the length of the corresponding periodic orbit. It does depend ``randomly'' on $N$ and, for large $N$, there exist ``short'' periodic orbits, for which $T[N]<<N.$ It is known, for the case $n=1$~\cite{falk_dyson} and ${\sf M}$ the Arnol'd cat map, that the smallest period corresponds to  $N$ a Fibonacci integer, $f_q$. In this case, $T[f_q]=2q$. 

What is significant about the period, $T[N]$, of the--hyperbolic--map, ${\sf M},$ is that it controls the rate of ``spreading'', with time,  of a localized distribution of initial conditions and, due to the compactness of the phase space, the rate of mixing~\cite{falk_dyson}. Although, for finite $N$, there isn't any rigorous notion of ergodicity and of mixing, we can get important information, for reasonably large values of $N$, about the mixing time, by considering evolution times equal to half of the period. In the case of the Arnol'd cat map, for $N$ a Fibonacci integer, the mixing time, $t_\mathrm{mixing}$ scales as $\log\,N$. 

It is this property that implies that the dynamics of the Arnol'd cat map is chaotic and that it is possible to acquire information about the chaotic behavior 
  by studying  orbits whose period takes  ``large'' values. The reason is that it is known that  chaotic orbits, whether for Hamiltonian or non--Hamiltonian systems, essentially can be identified as unstable periodic orbits of {\em infinite} period~\cite{cvitanovic2005chaos,gutkin:2019sca}. For chaotic systems, described by the dynamics of symplectic linear maps, i.e. elements of Sp$_{2n}[\mathbb{Z}]$ and phase space $\mathbb{T}^{2n}$, these linear maps must be hyperbolic, i.e. their eigenvalues must be  real and positive. 

One consequence of the above  is that  these eigenvalues come in pairs, $(\lambda,1/\lambda)$, with $\lambda > 1$. These pairs define planes in the $2n-$dimensional phase space, spanned by the corresponding eigenvectors, where the flow expands along the eigenvector, corresponding to the eigenvalue $\lambda$ and contracts along the eigenvector, corresponding to the eigenvalue $1/\lambda$.

Closing this section we recall, for completeness, properties of the finite group Sp$_{2n}[\mathbb{Z}_N]$.

The mod $N$ reduction of the symplectic group, $\mathrm{Sp}_{2n}[\mathbb{Z}]$, defines the finite symplectic group, $\mathrm{Sp}_{2n}[\mathbb{Z}_N]$ as follows: 
\begin{equation}
\label{Sp2nmodN}
\mathrm{Sp}_{2n}[\mathbb{Z}_N]=\left\{
{\sf M}\,\mathrm{mod}\,N,\forall {\sf M}\in \mathrm{Sp}_{2n}[\mathbb{Z}]
\right\}
\end{equation}
The mod $N$ reduction is a homomorphism from $\mathrm{Sp}_{2n}[\mathbb{Z}]$ to $\mathrm{Sp}_{2n}[\mathbb{Z}_N]$, with kernel the principal congruent subgroup,
\begin{equation}
\label{GammaN}
\Gamma_\mathrm{sp}[N]=\left\{
{\sf M}\in\mathrm{Sp}_{2n}[\mathbb{Z}]\left| {\sf M}\equiv\,I_{2n\times 2n}\,\mathrm{mod}\,N\right.
\right\}
\end{equation}

We can describe  the decomposition of the finite group $\mathrm{Sp}_{2n}[\mathbb{Z}_N]$ into the direct product of finite groups of the form $\mathrm{Sp}_{2n}[\mathbb{Z}_{p^k}]$, where $p$ is a prime and $k$ is a positive integer, corresponding to the decomposition in prime factors of $N$
\begin{equation}
\label{SpNpi}
N=\prod_{l=1}^L\,p_l^{k_l}
\end{equation}
Indeed, using the Chinese Remainder Theorem~\cite{Athanasiu:1998cq}, we can show that
\begin{equation}
\label{ChRemThSp}
\mathrm{Sp}_{2n}[\mathbb{Z}_N]=
\bigotimes_{l=1}^L\,\mathrm{Sp}_{2n}\left[\mathbb{Z}_{p_l^{k_l}}\right]
\end{equation}
This allows to obtain the order of the group, $\mathrm{Sp}_{2n}[\mathbb{Z}_N]$, since the order of each term of this decomposition is known, namely~\cite{sp2norder},
\begin{equation}
\label{orderSp2np}
\mathrm{ord}\left(\mathrm{Sp}_{2n}[\mathbb{Z}_{p^k}]\right)=p^{ (2k-1)n^2+(k-1)n }\prod_{i=1}^n\,\left(p^{2i}-1\right)
\end{equation}

.
\newpage
\section{Interacting Arnol'd cat maps from symplectic  couplings of $n$, $k-$Fibonacci sequences}\label{symplfibseq} 
In the previous section we defined  evolution operators $\mathrm{Sp}_{2n}[\mathbb{Z}]\ni{\sf M}:\mathbb{T}^{2n}\to\mathbb{T}^{2n},$ that act on the points of the $2n-$torus.  These evolution operators describe the dynamics in phase space of $n$ oscillators, each defined on a lattice of $n$ sites. However this construction doesn't show how the oscillators are actually coupled. So we must show how it is possible to obtain the evolution operator of $n$ oscillators, in terms of the evolution operator of one oscillator, how the phase space of one is embedded in the phase space of all. 

This is the subject of the present section. 
 
We shall show how to couple $n$ Arnol'd cat maps. We know that the evolution operator for each is an element of Sp$_2[\mathbb{Z}],$ the set of rational points of a 2-torus, $\mathbb{T}^2.$  The total phase space of the system will be $\mathbb{T}^{2n}$, the $2n-$dimensional torus and the proposed dynamics will be described by appropriate elements of the symplectic group, Sp$_{2n}[\mathbb{Z}]$. 

The idea is to exploit the known  correspondence between the Arnol'd cat map and the Fibonacci sequence and describe the coupling between the Arnol'd cat maps by the coupling between the sequences, so that the evolution operators of $n$ coupled Arnol'd cat maps can be understood as  iteration matrices of $n$ coupled {\em generalized} Fibonacci sequences.

Coupled Fibonacci sequences have been considered in the literature, for instance in~\cite{atanassov1985new,singh2010coupled,rathore2012generalized}. However, in these papers the possible applications to Hamiltonian dynamics  were not the topic of interest and, moreover, the corresponding maps were not symplectic.

In the literature it has been a matter of debate, how to couple together Arnol'd cat maps. What we propose in this paper is to define the coupling between Arnol'd cat maps, through the coupling between generalized Fibonacci sequences, by requiring that the resulting map be symplectic. 

To understand how this works, let us show how this works for two maps. 

We shall {\em define} the coupling between two  $n=2$  cat maps using the coupling between    $n=2$ {\em generalized}  Fibonacci sequences, $\{f_m\}$ and $\{g_m\}$ (but we write the expressions in a way that generalizes immediately to arbitrary $n$), as follows:
\begin{equation}
\label{Fibonacci}
\begin{array}{l}
\displaystyle
f_{m+1}=a_1f_m+b_1f_{m-1}+c_1g_m+d_1g_{m-1}\\
\displaystyle
g_{m+1}=a_2g_m+b_2g_{m-1}+c_2f_m+d_2f_{m-1}\\
\end{array}
\end{equation}
where the $a_i,b_i,c_i,d_i$ are integers,  $f_0=0=g_0$ and $f_1=1=g_1$ are the initial conditions and $m=1,2,3,\ldots$. A byproduct of our analysis will be how to  define the coupling between  $k-$Fibonacci sequences (in particular the case $k=1,$ which corresponds to the case of two Arnol'd cat maps).

Based on the properties of symplectic matrices, that we reviewed in the previous section, we write these equations in the following  matrix form 
\begin{equation}
\label{2Fibiter}
X_{m+1}\equiv
\left(\begin{array}{l} f_m\\g_m\\f_{m+1}\\g_{m+1}\end{array}\right)=
\left(\begin{array}{cccc}
0 & 0 & 1 &0\\
0 & 0 & 0 &1\\
b_1 & d_1 & a_1 & c_1 \\
d_2 & b_2 & c_2 & a_2
\end{array}
\right)
\underbrace{
\left(\begin{array}{l} f_{m-1}\\g_{m-1}\\f_{m}\\g_{m}\end{array}\right)}_{X_m}
\end{equation}
In this expression we now focus on  the 2$\times$2 matrices
\begin{equation}
\label{blocks}
\begin{array}{ccc}
{\sf D}\equiv\left(\begin{array}{cc} b_1 & d_1 \\ d_2 & b_2 \end{array}\right) & & 
{\sf C}\equiv\left(\begin{array}{cc} a_1 & c_1 \\ c_2 & a_2\end{array}\right)
\end{array}
\end{equation}
in terms of which the one--time--step evolution equation~(\ref{2Fibiter})  can be written in block form as 
\begin{equation}
\label{block1}
X_{m+1}=
\left(\begin{array}{cc} 0_{n\times n} & I_{n\times n} \\ {\sf D} & {\sf C}\end{array}\right) X_m
\end{equation}
In analogy with the case of a single Fibonacci sequence and its relation with the Arnol'd cat map, we impose the constraint (cf. ~(\ref{nonsymplA})) 
\begin{equation}
\label{symplectic1}
\left(\begin{array}{cc} 0_{n\times n} & I_{n\times n} \\ {\sf D} & {\sf C}\end{array}\right)^\mathrm{T}
{\sf J}
\left(\begin{array}{cc} 0_{n\times n} & I_{n\times n} \\ {\sf D} & {\sf C}\end{array}\right) = -{\sf J}
\end{equation}
This condition implies that 
\begin{equation}
\label{solsAB}
\begin{array}{ccc}
{\sf D} = I_{n\times n} & & 
{\sf C} = {\sf C}^\mathrm{T}
\end{array}
\end{equation}
Therefore $a_1=k_1$, $a_2=k_2$, $c_1 = c_2 = c$.

 In terms of these parameters, the recursion relations take the form
\begin{equation}
\label{New2Fib}
\begin{array}{l}
\displaystyle
f_{m+1}=k_1f_m+f_{m-1}+cg_m\\
\displaystyle
g_{m+1}=k_2g_m+g_{m-1}+cf_m\\
\end{array}
\end{equation}
and can be identified as describing a particular  coupling between a $k_1-$ and a $k_2-$Fibonacci sequence (cf. appendix~\ref{appkfibseq}). This particular coupling is determined by the condition that the square of the evolution matrix is an element of Sp$_4[\mathbb{Z}]$:
\begin{equation}
\label{Msquareevol}
{\sf A}=\left(\begin{array}{cc} 
0 & 1 \\ 1 & {\sf C}
\end{array}\right)\Rightarrow
{\sf M}={\sf A}^2=
\left(\begin{array}{cc} 
1 & {\sf C} \\ {\sf C} & 1+{\sf C}^2
\end{array}\right)
\end{equation}

The role of the coupling is played by the integer $c$.
In components: 
\begin{equation}
\left(\begin{array}{cccc}
0 & 0 & 1 &0\\
0 & 0 & 0 &1\\
1 & 0 & k_1 & c \\
0 & 1 & c & k_2
\end{array}
\right)^2 =
\left(\begin{array}{cccc}
1 & 0 & k_1 & c \\
 0 & 1 & c & k_2 \\
 k_1 & c & c^2+k_1^2+1 & c (k_1+ k_2) \\
 c & k_2 & c (k_1+k_2) & c^2+k_2^2+1 \\
\end{array}\right)
\end{equation}
and represents the discrete time evolution matrix for the coupling of two ``generalized'' Arnol'd cat maps.  

We remark that, if $c=0$ and $k_1=k_2=1,$ we recover two, decoupled, Arnol'd cat maps; if $c\neq 0,$ and $k_1=k_2=1$ we can, thereby, identify two ``coupled'' Arnol'd cat maps, while, if $k_1=k_2=k$, the system decouples into two, independent, $(k+c)$, resp. $(k-c)$ cat maps, for $f_m\pm g_m.$

The generalization to $n$ cat maps proceeds as follows:  We choose two diagonal matrices, of positive integers, ${\sf K}_{IJ}=K_I\delta_{IJ}$ and ${\sf G}_{IJ}=G_I\delta_{IJ}$, with $I,J=1,2,\ldots,n$. If $K_I=1,$ the cat maps are Arnol'd cat maps, with their coupling defined by the vector $G_I.$

One way to write the coupling between the maps is, once more, to work with the (generalized) Fibonacci sequences. To this end, we define the translation operator, ${\sf P}$ and, for $n>2,$ we can distinguish between a ``closed'' and an ``open''  chain, of maps, by  defining ${\sf P}_{I,J}=\delta_{I-1,J}\,\mathrm{mod}\,n$ for the former (and setting ${\sf P}_{1,n}=0={\sf P}_{n,1}$ for the latter).   The periodicity is expressed by the fact that 
${\sf P}^n=I_{n\times n}$. 
Moreover, ${\sf P}$ is orthogonal, since ${\sf P}{\sf P}^\mathrm{T}=I_{n\times n}$.  

Now we can define the coupling matrix for  $n$  sequences as
\begin{equation}
\label{evolAn}
{\sf C} = {\sf K} + {\sf P}{\sf G} + {\sf G}{\sf P}^\mathrm{T}
\end{equation}

The corresponding $2n\times 2n$ evolution matrix, ${\sf A}$ is given by 
\begin{equation}
\label{Aevol2n}
{\sf A} = \left(\begin{array}{cc} 
0_{n\times n} & I_{n\times n} \\ I_{n\times n} & {\sf C}
\end{array}\right)
\end{equation}
and satisfies the relation ${\sf A}^\mathrm{T}{\sf J}{\sf A}=-{\sf J}$.
Its square, 
\begin{equation}
\label{Mevol2n}
{\sf M}={\sf A}^2=\left(\begin{array}{cc} 
I_{n\times n} & {\sf C} \\ {\sf C} & I_{n\times n}+{\sf C}^2
\end{array}\right)
\end{equation}
This satisfies the relation  ${\sf M}^\mathrm{T}{\sf J}{\sf M}={\sf J}$, showing that 
${\sf M}\in\mathrm{Sp}_{2n}[\mathbb{Z}]$. 

Since ${\sf A}$ is symmetric, (from the property that ${\sf C}={\sf C}^\mathrm{T}$), ${\sf M}$ is positive definite and its eigenvalues come in pairs, $(\lambda,1/\lambda)$, with $\lambda > 1$. This property implies that, for all matrices ${\sf K}$ and ${\sf G}$ this system of coupled maps is hyperbolic. 

It is possible to decompose the classical evolution matrix ${\sf M}$  in terms of the generators of the symplectic group~(\ref{sympldecomp})
\begin{equation}
\label{Mdecomp}
{\sf M} = \left(\begin{array}{cc} I_{n\times n} & 0_{n\times n} \\ {\sf C} & I_{n\times n}\end{array}\right)\left(\begin{array}{cc} I_{n\times n} & {\sf C} \\ 0_{n\times n} & I_{n\times n}\end{array}\right)
\end{equation} 
Moreover each factor generates, for any symmetric, integer, matrix ${\sf C}$, an abelian subgroup of $\mathrm{Sp}_{2n}[\mathbb{Z}]$. These factors are called ``left'' (resp. ``right'') translations.

Eq.~(\ref{Mevol2n}) is, indeed, a key result of  our paper, since it shows how the evolution operator of $n$ cat maps is defined from the evolution operator of the individual maps, in a way consistent with symplectic covariance.

If ${\sf C}=I_{n\times n},$ we have $n$ decoupled Arnol'd cat maps, while the off-diagonal elements of ${\sf C}$ describe their interaction. If ${\sf C}$ is diagonal, we have decoupled cat maps and the band structure of ${\sf C}$ encodes the (non-)locality of the interactions. (This will become clearer when we shall discuss the dynamics in configuration space.)

An important special case arises if we impose  translation invariance along the chain of maps, i.e.  $K_I=K$ and $G_I = G$ for all $I=1,2,\ldots,n$. 

As an  example of the translation invariant closed chain, we present below  the matrix ${\sf C}$, for $n=3$:
\begin{equation}
\label{C4}
{\sf C}=\left(\begin{array}{ccc}
K & G & G\\
G & K  & G \\
G &  G & K
\end{array}
\right)
\end{equation}
The corresponding evolution map, ${\sf A}$, that describes three, coupled, $K-$Arnol'd cat maps, is a $6\times 6$ matrix, given by the expression
\begin{equation}
\label{Aarnold3}
{\sf A} = \left(\begin{array}{cccccc}
 0 & 0 & 0  & 1 & 0 & 0 \\
 0 & 0 & 0  & 0 & 1 & 0 \\
 0 & 0 & 0  & 0 & 0 & 1  \\
 1 & 0 & 0 & K & G &  G \\
 0 & 1 & 0  & G & K & G \\
 0 & 0 & 1  & G & G & K \\
\end{array}\right)
\end{equation} 
and the symplectic map, ${\sf M}={\sf A}^2$, of three, coupled $K-$Arnol'd cat maps, is given by the expression 
\begin{equation}
\label{Marnold3}
{\sf M}=\left(
\begin{array}{cccccc}
 1 & 0 & 0 & K & G & G \\
 0 & 1 & 0 & G & K & G \\
 0 & 0 & 1 & G & G & K \\
 K & G & G & 2 G^2+K^2+1 & G^2+2 G K &
   G^2+2 G K \\
 G & K & G & G^2+2 G K & 2 G^2+K^2+1 &
   G^2+2 G K \\
 G & G & K & G^2+2 G K & G^2+2 G K & 2
   G^2+K^2+1 \\
\end{array}
\right)
\end{equation}
It is interesting that the coupling $G$ appears, not only, in the off-diagonal $3\times 3$ blocks, but, also, in the diagonal elements of the lower block. On the other hand, setting $G=0$ we recover the case of three, decoupled ``$k-$Arnol'd''  cat maps. 

Let us now consider the case of the open chain. The only change involves the operator ${\sf P}$, which, now, must be defined as ${\sf P}_{IJ}=\delta_{I-1,J}$, for $I,J=1,2,\ldots,n$. Due to the absence of the mod $n$ operation, the ``far non--diagonal'' (upper right and lower left) elements are, now, zero. This express the property that the $n-$th Fibonacci is not coupled to the first one (and vice versa). 

For both, closed or open, chains, we observe certain algebraic properties of the evolution matrix, ${\sf A}$. 

The $k-$Fibonacci sequence has the important property that the elements of the matrix  ${\sf A}(k)^m$ are arranged in columns of consecutive pairs of the sequence. We shall show that this property can be generalized for $n$ interacting $k-$Fibonacci sequences as follows:

\begin{theorem}
The $m-$th power of the evolution matrix, ${\sf A}$ (cf. eq.~(\ref{Aevol2n})) can be written as 
\begin{equation}
\label{Amevol2n}
{\sf A}^m=\left(\begin{array}{cc} {\sf C}_{m-1} & {\sf C}_m\\ {\sf C}_m & {\sf C}_{m+1}\end{array}\right) 
\end{equation}
where ${\sf C}_0=0_{n\times n}$, ${\sf C}_1=I_{n\times n}$ and ${\sf C}_{m+1}={\sf C}{\sf C}_m+{\sf C}_{m-1}$, with $m=1,2,3,\ldots$.   
This matrix recursion relation generalizes to matrices the $k-$Fibonacci sequence for numbers, holds for any  matrix, ${\sf C}$ and, in particular for the (symmetric, integer) matrix ${\sf C}$, defined by eq.~(\ref{evolAn}).
\end{theorem}
\begin{proof}
The proof is by induction. For $m=1$ it is true, by definition. If we assume it holds for $m>1$, then, by the relation ${\sf A}^{m+1}={\sf A}\cdot{\sf A}^m$, we immediately establish that it holds for $m+1$. 
\end{proof}
It is straightforward to generalize this construction,  to take into account interactions between next-to--nearest neighbors, and so on. For the case of the closed chain, the most general construction is encoded in the matrix ${\sf C}$. Its definition~(\ref{evolAn}) can be written as 
\begin{equation}
\label{AevolAngen}
{\sf C} = {\sf K}I_{n\times n} + \sum_{l=1}^{[n/2]-1} \left( {\sf P}^l{\sf G}_l + {\sf G}_l\left[{\sf P}^\mathrm{T}\right]^l\right)
\end{equation}
The label $l$ refers to the neighborhood: $l=1$ labels the nearest neighbors, $l=2$ the next-to-nearest neighbors and so on. The farthest neighbors, on the closed chain, are $[n/2]-1$ sites apart. 

The matrices ${\sf G}_l$ are all diagonal, with integer entries and represent the couplings between the maps, at different sites of the lattice. By construction, the matrix ${\sf C}$ is symmetric, therefore the evolution matrix ${\sf M}\in\mathrm{Sp}_{2n}[\mathbb{Z}]$. 

Imposing, once more, translation invariance, the matrices ${\sf G}_l=G_l I_{n\times n}$, therefore, ${\sf C}$ becomes
\begin{equation}
\label{Ctransinv}
{\sf C}= KI_{n\times n} + \sum_{l=1}^{[n/2]-1}\,G_l\left( {\sf P}^l+\left[{\sf P}^\mathrm{T}\right]^l\right)
\end{equation}
The chaotic properties of the corresponding matrix ${\sf M}$, depend strongly on how $G_l$ depends on $l$; namely, whether $G_l$ decreases, is independent of, or increases with $l$. 

As an illustration, we show the matrix ${\sf C}$ for $n=7$, that can describe couplings up to third nearest neighbors:
\begin{equation}
\label{Apinakas7}
{\sf C}=\left(
\begin{array}{ccccccc}
  K & G_1 & G_2 & G_3 & G_3 & G_2 & G_1 \\
   G_1 & K & G_1 & G_2 & G_3 & G_3 & G_2 \\
  G_2 & G_1 & K & G_1 & G_2 & G_3 & G_3 \\
 G_3 & G_2 & G_1 & K & G_1 & G_2 & G_3 \\
 G_3 & G_3 & G_2 & G_1 & K & G_1 & G_2 \\
  G_2 & G_3 & G_3 & G_2 & G_1 & K & G_1 \\
 G_1 & G_2 & G_3 & G_3 & G_2 & G_1 & K \\
\end{array}\right)
\end{equation}
Here  we have assumed translation invariance,  ${\sf K}=K I_{n\times n}$ and we remark that all pairs of symmetric diagonals contain identical elements, $G_1, G_2, \ldots,G_{[n/2]-1}=G_3$ (for $n=7$).

Another direction involves considering higher dimensional lattices of coupled Arnol'd cat maps. 

For example  the case of the square lattice (with periodic boundary conditions) can be described in the following way: 

Let $f_m^{(I,J)}$ be the family of sequences, at timestep $m,$ where $I, J=0,1,2,\ldots,n-1$. Therefore we have $n^2$ Fibonacci sequences. The neighbors of the site $(I,J)$ are taken as $(I\pm1,J)$ and $(I,J\pm1)$.  The translation operators, that connect any site with its neighbors, are $P\otimes I_{n\times n}, P^\mathrm{T}\otimes I_{n\times n}, I_{n\times n}\otimes P$ and $I_{n\times n}\otimes P^\mathrm{T}$. It is easy to convince oneself that these translation operators determine the order of the $n^2$ Fibonacci sequences along a vector of length $n^2$. More explicitly, the ordering is the ``lexicographic'' ordering. The second index, $J$,  of $f_n^{(I,J)}$ is the ``fast'' index, while the index $I$ is the ``slow'' index. In row form the ordering is the following: $(f_m^{(0,0)},f_m^{(0,1)},\ldots,f_m^{(0,n-1)},f_m^{(1,0)},f_m^{(1,1)},\ldots,f_m^{(1,n-1)},\ldots,f_m^{(n-1,0)},f_m^{(n-1,1)},\ldots,f_m^{(n-1,n-1)})$.

The corresponding matrix ${\sf C}$, which encodes the couplings between nearest neighbors, contains two,  diagonal, $n^2\times n^2$,  matrices, one of which is 
${\sf K}$, just as for the case of the chain, along with another  matrix ${\sf G}$, which contains all the nearest--neighbor couplings. 

Schematically, 
\begin{equation}
\label{Csqulatt}
{\sf C} = {\sf K} + \left( P\otimes I + I\otimes P\right){\sf G} + {\sf G}\left(P^\mathrm{T}\otimes I + I\otimes P^\mathrm{T}\right)
\end{equation} 
In analogy with the one--dimensional case, interactions involving larger neighborhoods can be described by replacing $P,$ respectively $P^\mathrm{T},$ by $P^l$, respectively $[P^\mathrm{T}]^l.$

 Higher dimensional (hypercubic) lattices of Arnol'd cat maps can be described in the same way.

In summary we have constructed the evolution operator for $n$ Arnol'd cat maps in a way that is consistent with its action as an element of Sp$_{2n}[\mathbb{Z}]$ on the torus $\mathbb{T}^{2n}$ and have shown how it is built up from the evolution operator of the individual maps. 

If the coordinates of the initial condition are rational, then, as we have explained, the mod 1 operation, which expresses the fact that the action takes place on the torus, is replaced by the mod $N$ operation, where $N$ is the least common multiple of the denominators of the coordinates. These symplectic maps are all elements of Sp$_{2n}[\mathbb{Z}]$, since the matrices ${\sf K}$ and ${\sf G}_l$ are all integer--valued. Applying the restriction of the mod $N$ operation, these maps belong to the group Sp$_{2n}[\mathbb{Z}_N]$ and act on the toroidal lattice $\mathbb{T}^{2n}[N]$.  As noted before, all the orbits will be periodic, with period $T[N]$ of the corresponding map ${\sf M}$~(\ref{Mevol2n}). 

The next step in the study of the dynamics of the map, ${\sf M},$ entails computing its spectrum--from which we can deduce the Lyapunov exponents and hence the Kolmogorov--Sinai entropy--and its eigenvectors. This calculation is facilitated by studying the equations of motion in configuration space, where, indeed, locality makes more sense than in phase space.

In the next section, therefore, we shall construct, explicitly, starting from Hamilton's equations, the corresponding Newton's equations, which describe the discrete time evolution of the position variables, as well as their solutions. 
For   the case of   translation invariant couplings, i.e. $K$ and $G$ constant  (first, for nearest--neighbor interactions,  and,  subsequently, for any range $1 < L \leq [n/2]-1$), they take into account  the degree of locality of the interactions through their dependence on $L,$  the number of interacting neighbors.  $L=1$ means nearest-neighbor interactions and so on.
A particularly striking property of the equations in configuration space is that these  $n$ coupled maps,   describe a system of $n$ coupled inverted harmonic oscillators that don't exhibit  runaway behavior, since this is  ``cured''  by the compactness of the phase space. 

\newpage
\section{From Hamilton's to Newton's discrete time equations}\label{eom}
The classical equations of motion
\begin{equation}
\label{Meom}
\bm{x}_{m+1}=\bm{x}_m{\sf M}
\end{equation}
where $m$ is the iteration timestep of the map, 
can be written in terms of positions, $\bm{q}_m$ and momenta, $\bm{p}_m$ as
\begin{equation}
\label{Meomqp}
\begin{array}{l}
\displaystyle
\left(\bm{q}_{m+1},\bm{p}_{m+1}\right)=\left(\bm{q}_{m},\bm{p}_{m}\right)\left(\begin{array}{cc} 1 & {\sf C}\\ {\sf C} & 1 +{\sf C}^2\end{array}\right)\\
\displaystyle
\bm{q}_{m+1} = \bm{q}_m+ \bm{p}_m {\sf C}\\
\displaystyle
\bm{p}_{m+1} = \bm{q}_m{\sf C} + \bm{p}_m\left(1+{\sf C}^2\right) 
\end{array} 
\end{equation}
Since ${\sf C}$ is symmetric, we can unclutter notation considerably by omitting the ``transpose'' symbol. Henceforth the row vectors are written without it:

We may solve the first of the last two equations for $\bm{p}_m$, 
\begin{equation}
\label{pm}
\bm{p}_m = \bm{q}_{m+1}{\sf C}^{-1}-\bm{q}_m{\sf C}^{-1}\Leftrightarrow 
\bm{p}_{m+1}= 
\bm{q}_{m+2}{\sf C}^{-1}-\bm{q}_{m+1}{\sf C}^{-1}
\end{equation}
and insert the result in the second, in order to obtain a recursion relation for $\bm{q}_m$ only--i.e. Newton's equations of motion:
\begin{equation}
\label{Newtonslaw}
\bm{q}_{m+1}-2\bm{q}_m+\bm{q}_{m-1}= \bm{q}_m{\sf C}^2
\end{equation}
This equation describes the discrete time evolution of $n$ coupled Arnol'd cat maps, using only the coordinates in position space. It highlights that the locality properties of the system are encoded in the band structure of the symmetric matrix ${\sf C}.$

This procedure assumes that ${\sf C}$ is invertible; which fails to hold when ${\sf C}$ has a zeromode, that corresponds to a conserved quantity. 
In this case, the evolution matrix, ${\sf M},$ has an eigenvalue equal to 1 (for each zeromode). This, in turn, implies that, if we choose as initial conditions the zeromode itself, this will not evolve in time.

However we do not need to assume that ${\sf C}$ is invertible, in order to obtain eqs.~(\ref{Newtonslaw}).

In the following, we shall consider the case of the non--zero modes and discuss the zeromodes separately.   

In the subspace of the non-zero modes, the matrix ${\sf C}^2$ is, by construction, positive definite and eq.~(\ref{Newtonslaw})  describes coupled {\em inverted}  harmonic oscillators--the coupling is repulsive. The interest for this system stems from the fact that the phase space of each of these particles is a two--dimensional torus, so the motion is strongly chaotic and mixing (cf. also~\cite{de2022linear}).   

We now decouple the modes of Newton's equations and diagonalize ${\sf C}$ by (finite) Fourier transform, ${\sf F}^\dagger{\sf C}{\sf F}\equiv{\sf D}$ where ${\sf F}_{IJ}=e^{2\pi\mathrm{i}IJ/n}/\sqrt{n}\equiv \omega_n^{IJ}/\sqrt{n}.$ 
We define the mode variable $\bm{r}_m$ by $\bm{q}_m\equiv\bm{r}_m{\sf F}$. The mode variable $\bm{r}_m$ satisfies Newton's equation of motion in the form: 
\begin{equation}
\label{modes}
\bm{r}_{m+1}-2\bm{r}_m+\bm{r}_{m-1} = \bm{r}_m{\sf D}^2
\end{equation} 
where ${\sf D}_{IJ}=\delta_{IJ}D_J$, with 
\begin{equation}
\label{diagonalC}
D_{J}=K + 2G\cos\frac{2\pi J}{n}
\end{equation}
for the case of nearest--neighbor interactions. We note here that, if $n$ is even, then the mode $J_0=n/2$ has zero eigenvalue, when $K=2G.$
If $n$ is odd, on the other hand, a zeromode cannot exist, because $K$ and $G$ are positive integers. 

It is possible to include the case of couplings beyond nearest--neighbors, i.e. $G_l$, with $1< l \leq (n-1)/2$, as follows:
\begin{equation}
\label{longrangeints}
D_J = K + 2\sum_{l=1}^{\frac{n-1}{2}}\,\left(
G_l\cos\frac{2\pi lJ}{n}
\right)
\end{equation} 

We shall now determine the discrete time evolution of the normal modes of the chain, by setting $(\bm{r}_{m})_I=r_{I,m}\equiv \delta_{IJ}\rho_I^m a_J$ (where $I,J=1,2,\ldots,n$).
We duly find a quadratic equation for  $\rho_I$:
\begin{equation}
\label{eigenvaluespropchain}
\rho_I^2-(2+{\sf D}_I^2)\rho_I+1=0\Leftrightarrow\rho_{\pm,I}=\frac{2+{\sf D}_I^2}{2}\pm\frac{|{\sf D}_I|}{2}\sqrt{{\sf D}_I^2+4}
\end{equation}
Thus the general solution, for $\bm{r}_m$ can be written as
\begin{equation}
\label{gensolchain}
\bm{r}_m = \bm{a}_+\rho_+^m + \bm{a}_-\rho_-^m
\end{equation}
and the solution, in terms of $\bm{q}_m$, is 
\begin{equation}
\label{qsol}
\bm{q}_m=\bm{r}_m{\sf F}
\end{equation}
The initial conditions are $(\bm{q}_0,\bm{p}_0)=(\bm{q}_0,\bm{q}_{1}{\sf C}^{-1}-\bm{q}_0{\sf C}^{-1})$  or, equivalently, $(\bm{q}_0,\bm{q}_1)$. 

We can express the coefficients, $\bm{a}_\pm$, in terms of the modes, $\bm{r}_0$ and $\bm{r}_1$ (where $\rho_\pm$ are the diagonal matrices, with elements $\rho_{I,\pm}$)
\begin{equation}
\label{apm}
\begin{array}{l}
\displaystyle
\bm{r}_0 = \bm{a}_+ +\bm{a}_-\\
\displaystyle
\bm{r}_1 = \bm{a}_+\rho_+ + \bm{a}_-\rho_-
\end{array}
\Leftrightarrow
\begin{array}{l}
\displaystyle
\bm{a}_+ =\left(\bm{r}_1- \bm{r}_0\rho_-\right)\left(\rho_+ - \rho_-\right)^{-1}\\
\displaystyle
\bm{a}_-=\left(\bm{r}_0\rho_+ -\bm{r}_1\right)\left(\rho_+  - \rho_-\right)^{-1}
\end{array}
\end{equation}
These equations, of course, only hold for the non-zero mode sector, since they degenerate for the zeromodes. 

For the zeromodes $(r_m)_I=(r_0)_I$ and, similarly, for the corresponding zeromodes of the momenta.  

By defining $\nu_\pm\equiv {\sf F}^\dagger\rho_\pm{\sf F}$, we obtain the exact solution of eq.~(\ref{Newtonslaw}), for the discrete time evolution of  the positions, $\bm{q}_m$ in the form
\begin{equation}
\label{positions}
\bm{q}_m=\left[  \left(\bm{q}_0\nu_+-\bm{q}_1 \right)\nu_+^m + \left(\bm{q}_1-\bm{q}_0\nu_+ \right)\nu_-^m\right]\left( \nu_+-\nu_-\right)
\end{equation}
We notice here that, for initial conditions, $(\bm{q}_0, \bm{q}_1)$ integer vectors, $\bm{q}_m$ will be integer, also, for all times.

This is obvious, since the coupling matrix, ${\sf C}$, appearing in eq.~(\ref{Newtonslaw}) is integer--valued. 

Studying the periodic trajectories, we restricted the initial conditions to be integer vectors--mod $N.$ In order to find, at any time step, $m,$ the position $\bm{q}_m$ inside the torus $\mathbb{T}^{2n}[\mathbb{Z}_N]$ we have to realize mod $N$ reduction in eq.~(\ref{positions}). The modes are coupled, in position space, through the matrices $\nu_\pm$. In components, these read
\begin{equation}
\label{modecoupling}
\left(\nu_\pm\right)_{IJ} =  {\sf F}^\dagger_{IM}\left(\rho_\pm\right)_{MK}{\sf F}_{KJ} = \frac{1}{n}\omega_n^{(J-I)K}\left(\rho_{\pm}\right)_{K}
\end{equation}
since $\rho_\pm$ are diagonal. 

We remark that $\nu_\pm$ are real, symmetric  and positive definite matrices since  $(\rho_\pm)_K=(\rho_\pm)_{n-K}$, for $K=0,1,2,\ldots,n-1$ and their product $\nu_+\nu_-=I_{n\times n}$ , since $\rho_+\rho_-=I_{n\times n}$. 

Having determined explicitly all the periodic orbits of the system, we shall now study the spectrum of the Lyapunov exponents.

\newpage
\newpage
\section{ Tuneable non-locality, Lyapunov spectra and K-S entropy}\label{Lyap}

In this section we shall obtain analytic expressions for the Lyapunov spectra and the Kolmogorov--Sinai entropy, based on the calculations of the previous sections  and we shall discuss
their significance  for   the mixing (scrambling) properties of the ACML chaotic systems.

Let us start with the spectrum of the Lyapunov exponents, $\lambda_I, I=0,1,2,\ldots,n-1$,
which characterize the spatio-temporal chaotic properties of the chain. 

We shall consider two cases:
\begin{enumerate}
\item The case of nearest neighbor (nn)  interactions, {\em viz.} when  $G_l=G$ for $l=1$ and 0 for $l>1.$ 
\item The case of longer range interactions, {\em viz.} when  $G_l=0$, for $l>L$,$L=2,3,..,[(n-1)/2].$
\end{enumerate}

In both cases, the Lyapunov exponents are defined by  cf.~(\ref{eigenvaluespropchain}),
\begin{equation}
\label{Lyapexps}
\lambda_{\pm,I} =\log\rho_{\pm,I}=\log\left\{ \frac{2+{\sf D}_I^2}{2}\pm\frac{|{\sf D}_I|}{2}\sqrt{{\sf D}_I^2+4}\right\}
\end{equation}
In general, the $D_J$ are given by the expression
\begin{equation}
\label{longrangeints1}
D_J = K + 2\sum_{l=1}^{[\frac{n-1}{2}]}\,\left(
G_l\cos\frac{2\pi lJ}{n}
\right)
\end{equation} 
In the first case, typical density  plots for the  spectra are shown in fig.~\ref{lambdakgn}. 
\begin{figure}[thp]
\begin{center}
\includegraphics[scale=0.5]{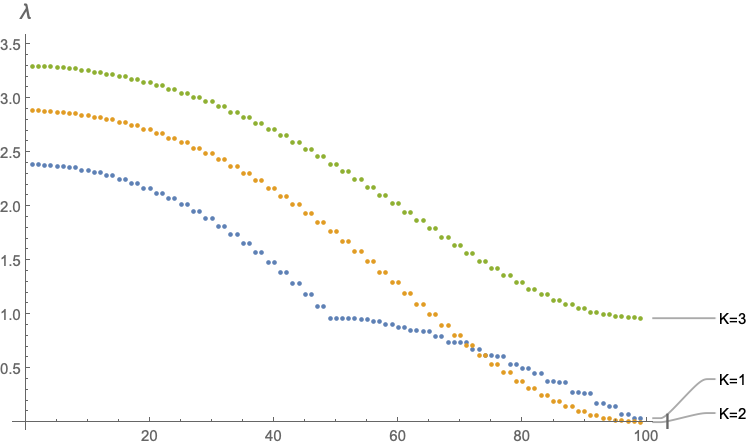}
\end{center}
\caption[]{Histogram of the sorted Lyapunov spectra, of $\lambda_{+,I}^{(K)}$ vs. $I=0,1,2,\ldots,n-1$, for $n=100$, $G=1,$ nearest-neighbor interactions--$L=1$--and $K=1,2,3.$
}
\label{lambdakgn}
\end{figure}

In the second case, we shall consider uniform couplings, {\em viz.} $G_l=G,$ for $l=1$ to $l=L<(n-1)/2$ and $G_l=0$ for $l>L.$

This particular choice is interesting for two reasons: First, we can compute the sum explicitly, {\em viz.} 
\begin{equation}
\label{Deval}
\begin{array}{l}
\displaystyle
D_J^{(L)}=K+2G\frac{\sin(JL\pi/n)}{\sin(J\pi/n)}\cos(\pi(L+1)J/n)
\displaystyle
\end{array}
\end{equation}
for $J=1,2,\ldots,n-1$ (and $1\leq L\leq (n-1)/2$). The case $J=0$ must be treated separately, since $D_0^{(L)}=K+2GL$. A typical example for the density plot  of the spectrum of the Lyapunov exponents in this case, is shown  in fig.~\ref{Lyapspectra}.

\begin{figure}[thp]
\begin{center}
\includegraphics[scale=0.5]{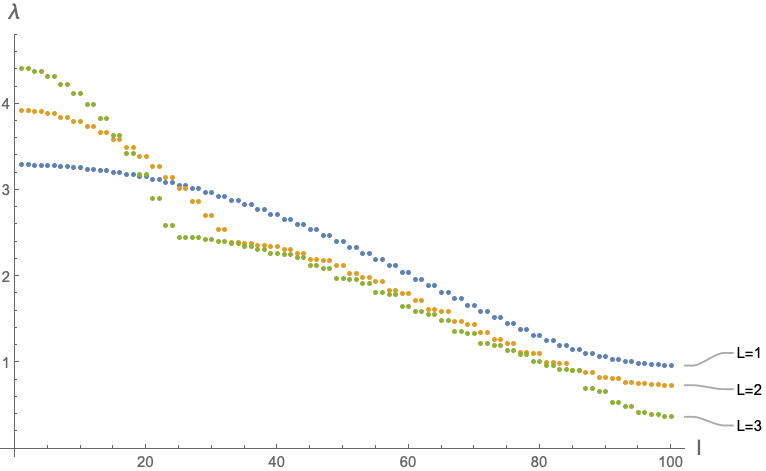}
\end{center}
\caption[]{Histogram of the sorted Lyapunov spectra, $\lambda_{+,I}^{(L)}$, vs. $I=0,1,2,\ldots,n-1$, for uniform couplings, namely, $K=3,G_l=G=1,n=101$ and $L=1,2,3.$
}
\label{Lyapspectra}
\end{figure}
An important consistency check of our calculations is that the sum of the positive Lyapunov exponents, for large values of $n,$ is a linear function of $n.$ 
Indeed this sum can be identified with the rate of entropy production, which is known as the Kolmogorov--Sinai entropy~\cite{pesin1977characteristic,sinai1989dynamical}.
 \begin{equation}
\label{KSentropyeq}
S_\mathrm{KS}=\sum_{I=0}^{n-1}\lambda_{+,I}
\end{equation}
We plot it, for the same values of $K$ and $G$ as above, for nearest--neighbor interactions, $L=1,$ as a function of the length of the chain, from $n=3$ to $n=100$, in fig.~\ref{KSentropy}. We remark that, for small sizes, there are deviations from linear behavior, that takes over at large sizes. 
\begin{figure}[thp]
\begin{center}
\includegraphics[scale=0.6]{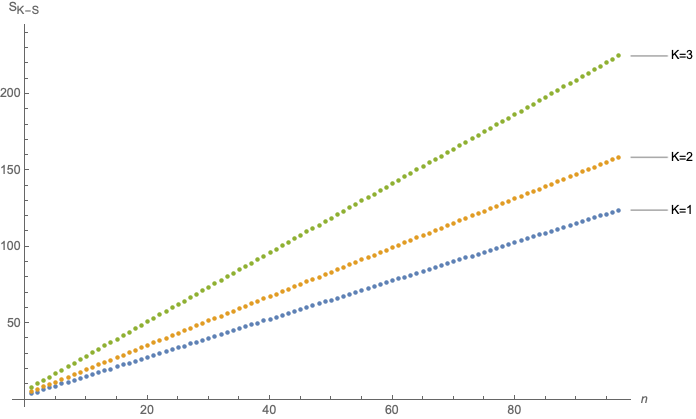}
\end{center}
\caption[]{The Kolmogorov--Sinai entropy, $S_\mathrm{KS}$, of eq.~(\ref{KSentropyeq}) as a function of the chain length, $n=5-101$, for $K=1,2,3$ and $G=1$. We remark deviations from linear behavior at small chain lengths, that become negligible at larger lengths.
}
\label{KSentropy}
\end{figure}
We observe that the slope of the K-S entropy, is an increasing function of $K.$

For longer-range, uniform interactions,  $L=5,$ the corresponding K-S entropy is displayed in fig.~\ref{KSentropyn=100}.
\begin{figure}[thp]
\begin{center}
\includegraphics[scale=0.6]{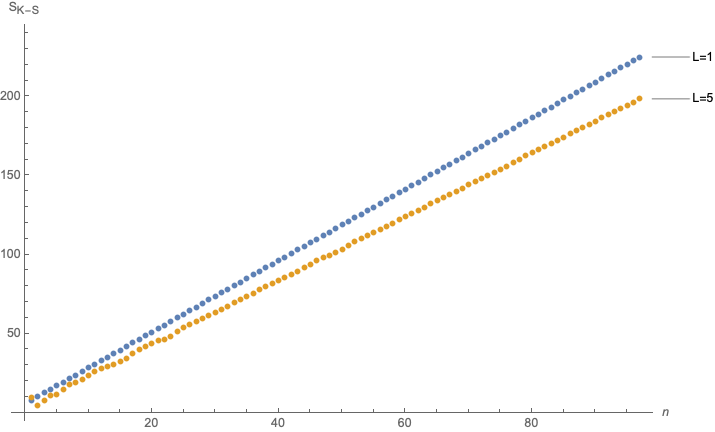}
\end{center}
\caption[]{Kolmogorov--Sinai entropy, for $n=5-101$ Arnol'd cat maps, $K=3,G=1$ with ranges $L=1$  and $L=5.$ 
} 
\label{KSentropyn=100}
\end{figure}

Let us discuss now the importance of the above typical behavior of the Lyapunov spectra and the K-S entropy. In the examples we considered above we notice the following four properties of the ACML systems.

First the increase of the average magnitude of the Lyapunov exponents as functions of the size of the system ,$n$,the coupling constant $G$ and the constant $K$ which plays a role for the mass of the coupled chaotic units, i.e.the Arnold cat maps.

Second the shape(curvature) of the Lyapunov spectra is tuneable, i.e we may control if we have many large or many small Lyapunov exponents or a flat region in the middle.

Third the increase of the slope of the K-S entropy as a function of $K$  of the system, for fixed $G$ and nearest-neighbor interactions. 

 Fourth--somewhat surprisingly--the long range interactions, generically, do {\em not}  increase the average magnitude of the Lyapunov exponents and consequently the K-S entropy--this can be understood as a result of the oscillations in the spectra of the Lyapunov exponents. 
  
 All these four properties are important for the tuneability of the mixing properties of the system.
 
 This system of $n-$coupled chaotic units (ACM) can be proven by general theorems that it is  mixing, because it is ergodic and it has a compact-bounded phase space, the $2 n$ dimensional torus.
 The mixing time of the system is defined as the logarithm of the deviation from the uniform distribution in time $T$ for an initially chosen probability distribution in the phase space,divided by the time $T,$ in the large $T$ limit. 
 
The problem of the calculation of the mixing time    is an interesting  exercise, whose solution, for the present system, will be described in detail in a future publication~\cite{AFNmixing}.
For the case of $n=1$ it is equal to the 1/logarithm of the golden ratio~\cite{arnold1968ergodic}. For $n>1$ it is expected to be proportional to $1/S_\mathrm{K-S}.$

In general we would expect that the mixing time is faster, the greater the K-S entropy, but there are known  counterexamples depending on the choice of the initial probability distribution~\cite{crawford1983decay}. 

For quantum systems there is the conjecture, as we discussed in the introduction, that the black holes are the fastest scramblers of the universe and their scrambling time is proportional to the logarithm of their entropy. However, whether this scrambling time can be identified with the mixing time of the quantum dynamical system is an open question and different scenaria have been proposed~\cite{latora1999kolmogorov}.

\newpage
\section{The spectra of periods of ACML systems }\label{permodN}

In this section we study the problem of finding  the spectrum of the periods of  the evolution  operator $M$  of $n$ coupled Arnol'd cat maps mod $N$. They are expected to be a random function of $N;$ however, for special values of $N,$ a thorough along the lines of Falk and Dyson~\cite{falk_dyson}, which has been done only for  single cat map--$n=1$--can only   lead to  bounds and, only in  some cases, to exact expressions.

In this section we shall present  an algorithm  for finding  the period of the evolution  operator of $n$ cat maps,  using properties of the matrix Fibonacci polynomials.

 Since the dynamics is that of a system of coupled, ``inverted'' harmonic oscillators--that are, however,  constrained to a (compact) toroidal phase space, $\mathbb{T}^{2n}[N]$--we expect that this system describes maximally chaotic and mixing motion.

To be concrete  consider the action of the evolution operator ${\sf M}$ (cf. eq.~(\ref{Mevol2n}) ), for the system of $n$ coupled Arnol'd cat maps, on the discrete phase space 
$\mathbb{T}^{2n}[N]$.  According to  Theorem~\ref{Amevol2n} the $m-$th power of ${\sf M}$ is given by the expression
\begin{equation}
\label{Mmevol2n}
{\sf M}^m={\sf A}^{2m}=\left(\begin{array}{cc} {\sf C}_{2m-1} & {\sf C}_{2m}\\ {\sf C}_{2m} & {\sf C}_{2m+1}\end{array}\right) 
\end{equation}
where ${\sf C}_0=0$, ${\sf C}_1=I_{n\times n}$ and ${\sf C}_{m+1}={\sf C}{\sf C}_m+{\sf C}_{m-1}$, with $m=1,2,3,\ldots$.   

Since ${\sf M}\in\mathrm{Sp}_{2n}[\mathbb{Z}]$, its mod $N$ reduction belongs to Sp$_{2n}[\mathbb{Z}_N]$. The order of the latter--finite--group can be determined using the relations~(\ref{ChRemThSp}) and~(\ref{orderSp2np}). These imply that the order of  the evolution matrix ${\sf M}$ mod $N$, $T(N)$, must be a divisor of the order of Sp$_{2n}[\mathbb{Z}_N]$.

In order to determine $T(N)$ we make use of Theorem~\ref{Amevol2n}  as follows: 

Since ${\sf M}^{T(N)}=I_{2n\times 2n}\,\mathrm{mod}\,N$, it is obvious   that  ${\sf C}_{2T(N)-1}\equiv I_{n\times n}\,\mathrm{mod}\,N$ and ${\sf C}_{2T(N)}\equiv 0\,\mathrm{mod}\,N,$ which reduces the problem of finding the period to finding the least value of $m=T(N),$ for which  these two relations hold simultaneously.

The period of ${\sf M}\,\mathrm{mod}\,N$ is the smallest integer, $T(N)$, such that 
\begin{equation}
\label{TofN}
{\sf M}^{T(N)}\equiv\,I_{2n\times 2n}\,\mathrm{mod}\,N
\end{equation}
This implies that 
\begin{equation}
\label{TofN1}
{\sf M}^{T(N)}=
\left(\begin{array}{cc} {\sf C}_{2T(N)-1} & {\sf C}_{2T(N)}\\ {\sf C}_{2T(N)} & {\sf C}_{2T(N)+1}\end{array}\right)\,\mathrm{mod}\,N\equiv
\left(\begin{array}{cc} I_{n\times n} & 0_{n\times n} \\ 0_{n\times n} & I_{n\times n}\end{array}\right)\,\mathrm{mod}\,N
\end{equation} 
Therefore 
${\sf C}_{2T(N)}\equiv 0\,\mathrm{mod}\,N$, 
${\sf C}_{2T(N)-1}\equiv I_{n\times n}\,\mathrm{mod}\,N$ and 
${\sf C}_{2T(N)+1}={\sf C}\cdot {\sf C}_{2T(N)}+{\sf C}_{2T(N)-1}\equiv\,I_{n\times n}\,\mathrm{mod}\,N$. 
From these relations it is easy to show that $2T(N)$ is the period of the sequence of  matrices $\{ {\sf C}_m\,\mathrm{mod}\,N\}$: 
\begin{proof}
The starting point is the property that 
\begin{equation}
\label{C2TN}
\begin{array}{l}
\displaystyle
{\sf C}_{2T(N)-1}\equiv\,I_{n\times n}\,\mathrm{mod}\,N\\
\displaystyle 
{\sf C}_{2T(N)}\equiv 0_{n\times n}\,\mathrm{mod}\,N
\end{array}
\end{equation}
This implies that 
\begin{equation}
\label{Cseq}
\begin{array}{l}
\displaystyle
{\sf C}_{2T(N)+1}={\sf C}\cdot{\sf C}_{2T(N)}+{\sf C}_{2T(N)-1}\equiv I_{n\times n}={\sf C}_1\,\mathrm{mod}\,N\\
\displaystyle
{\sf C}_{2T(N)+2}={\sf C}\cdot I_{n\times n}\equiv {\sf C}={\sf C}_2\,\mathrm{mod}\,N\\
\displaystyle
{\sf C}_{2T(N)+3}={\sf C}\cdot{\sf C}_{2T(N)+2}+{\sf C}_{2T(N)+1}\equiv I_{n\times n}+{\sf C}^2={\sf C}_3\,\mathrm{mod}\,N\\
\end{array}
\end{equation}
\end{proof}
Now the key observation is that ${\sf C}_m$ is a polynomial in the matrix ${\sf C}$, with positive  integer coefficients and its degree is equal to $m-1$, which is even for $m$ odd and odd for $m$ even. 

In fact these polynomials turn out to be nothing else but the so-called Fibonacci polynomials--now defined over the space of integer  matrices mod $N.$  

The Fibonacci polynomials are defined by the recursion relation~\cite{philippou2001fibonacci}
\begin{equation}
\label{fibpoly}
{\sf F}_{m+1}(x)=x {\sf F}_m(x)+{\sf F}_{m-1}(x)
\end{equation}
with $x$ a formal variable  and initial conditions ${\sf F}_0(x)=0$ and ${\sf F}_1(x)=1$. 

The Fibonacci polynomials have been extensively studied, for $x\in\mathbb{R}$; what we note here is that they can be defined for $x={\sf C}$, i.e. matrices; and many of their remarkable properties carry over to this case. 
In order, therefore,  to find the periods $\mathrm{mod}\,N$ we must find, for a given evolution matrix $C,$ the  values of $m,$ for which the matrix Fibonacci polynomial  ${\sf F}_{2T(N)}(x={\sf C}),$ vanishes and, simultaneously,  ${\sf F}_{2T(N)-1}(x={\sf C})=I_{n\times n} (\mathrm{mod}\, N).$

Using that  ${\sf C}_m={\sf F}_m({\sf C})$ and their  explicit formula we readily find that 
\begin{equation}
\label{mfibo}
{\sf F}_m({\sf C})=
\sum _{j=0}^{\left[\frac{m-1}{2}\right]} \binom{-j+m-1}{j} {\sf C}^{-2 j+m-1},
\end{equation}
where $[\cdot]$ denotes the integer part of the argument. 

The reason these polynomials are particularly useful here is that the evolution in phase space is described by the recursion relation
\begin{equation}
\label{evolphsp}
(\bm{q}_m,\bm{p}_m)=(\bm{q}_{m-1},\bm{p}_{m-1})\cdot {\sf M}=(\bm{q}_0,\bm{p}_0)\cdot {\sf M}^m=(\bm{q}_0,\bm{p}_0)
\left(\begin{array}{cc} {\sf F}_{2m-1}({\sf C}) & {\sf F}_{2m}({\sf C}) \\ {\sf F}_{2m}({\sf C})& {\sf F}_{2m+1}({\sf C})\end{array}\right)
\end{equation}
and we realize that the coefficients in the last expression are the Fibonacci polynomials for matrices--which allow us to write the evolution in phase space in closed form:
\begin{equation}
\label{evolphspsol}
\begin{array}{l}
\displaystyle
\bm{q}_m=\bm{q}_0 {\sf F}_{2m-1}({\sf C}) + \bm{p}_0 {\sf F}_{2m}({\sf C})\\
\displaystyle 
\bm{p}_m=\bm{q}_0 {\sf F}_{2m}({\sf C}) + \bm{p}_0 {\sf F}_{2m+1}({\sf C})\\ 
\end{array}
\end{equation}
These equations highlight that  at step  $m=T(N)$, which is defined by eqs.~(\ref{C2TN}), $\bm{q}_{T(N)}=\bm{q}_0$ and 
$\bm{p}_{T(N)}=\bm{p}_0,$ consistent with $T(N)$ being the period of motion.

Our framework provides a lot of freedom for choosing the dynamics:
\begin{enumerate}
\item The number of maps, $n.$
\item The positive integer, $K,$ from the $K-$Fibonacci sequence.
\item The range of the non-locality, $l;$ $l=1$ is for nearest neighbor interactions, $l=2$ for next--to--nearest neighbor interactions, and so on; $l\leq l_\mathrm{max}=\mathrm{integer\,part}(n-1)/2.$
\item The couplings, $G_l,$ for $l=1,2,\ldots,l_\mathrm{max},$ that are, also, positive integers. 
  \end{enumerate}
To simplify matters, we shall provide numerical examples for the periods, $T[N],$ by choosing all the couplings, $G_l\equiv G$ and for a given range of non--locality. Moreover we shall choose $N=p$ a prime and  $n=2-5$ (so up to five maps). 
(The case $n=1,$ of the one cat map, has been thoroughly studied in the past~\cite{Keating1991,Axenides:2016nmf}.)

For these choices, the expression for the order of the group simplifies considerably:
\begin{equation}
\label{orderSpN=p}
\mathrm{ord}\left(\mathrm{Sp}_{2n}[\mathbb{Z}_{p}]\right)=  
p^{ n^2 }\prod_{i=1}^n\,\left(p^{2i}-1\right)
\end{equation}

To conclude this section we report on  the results of the numerical investigations for $T(N)$ for selected values of $N$ and $K=G=1$
(cf. fig.~\ref{20primes})  
\begin{figure}[thp]
\begin{center}
\subfigure{\includegraphics[scale=0.5]{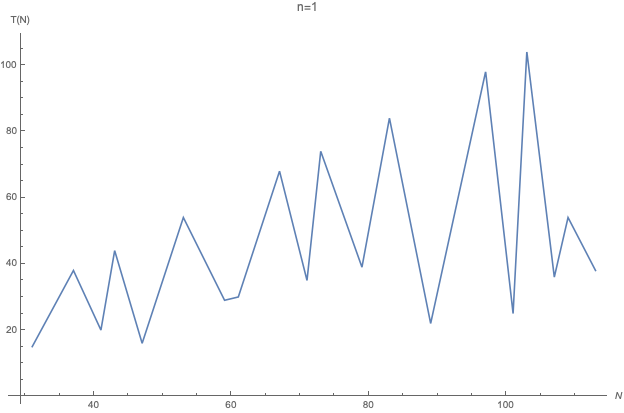}}
\subfigure{\includegraphics[scale=0.5]{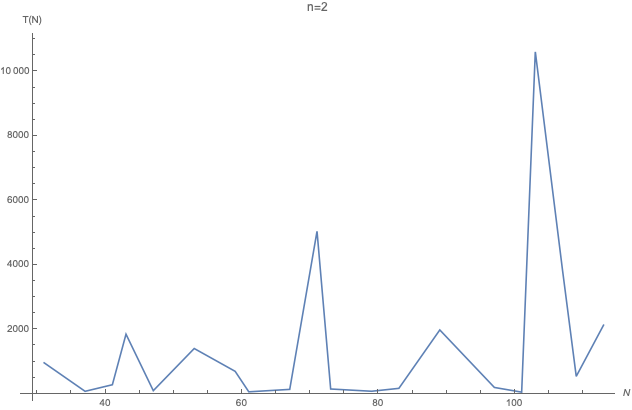}}
\end{center}
\caption[]{Period $T(N)$ for $N=p_{11}=31$ (the eleventh prime) to $p_{31}=113,$ (the thirty-first prime) for $l=1,$ $n=1$ and 2. We remark the dramatic change from $n=1,$ one map, to $n=2,$ two coupled maps. This reflects the dramatic increase in size of the order of the group.
} 
\label{20primes}
\end{figure}

\newpage
\section{Conclusions and outlook}\label{concl}
The study of classical and quantum chaotic field theories has received considerable attention recently,   inspired by work in turbulence,  the problem of statististical $n-$body thermalization (the so-called Eigenstate Thermalization Hypothesis),  as well as from the motivation to describe the quantum dynamics of black holes.  
To this end we started by setting forth the description of the classical  $n-$body lattice field theories, whose fundamental constituents are Arnol'd cat maps. The underlying quantum dynamics  will be the subject of future work.

More specifically, in this work, we have presented the consistent Hamiltonian dynamics of
  coupled map lattices of $n$ classical chaotic oscillators--Arnol'd cat maps (ACM)--in phase space and in  cofinguration space--with guiding principle the symplectic invariance of the phase space of the $n-$body system. 
  
  In configuration space, each  ACM is  located  on a single  site of the lattice  and acts on a two-dimensional toroidal phase space in a way that is   hyperbolic and exhibits  maximal mixing. 
  
Our method is based on the  representation of the  Arnol'd cat map in terms of Fibonacci sequences with the $n-$body coupled ACM generalization being realized  through the coupling of $n$  $k-$Fibonacci sequences.  

The   $n-$body system is thus defined  on a $2n-$dimensional toroidal phase space, by the action of  elements of the symplectic group  Sp$_{2n}[\mathbb{Z}]$, which is maximally  hyperbolic and mixing, thanks to the periodic boundary conditions. 

The corresponding equations of motion in configuration space are those of $n,$ linearly coupled,  inverted harmonic oscillators. It is interesting to stress that it is the boundary conditions in the phase space that ensure that the system doesn't have runaway behavior. This is an example of how mechanical systems,  with unbounded potentials, can be understood as chaotic systems, upon imposing periodic  boundary conditions--in phase space.

The chaotic properties of the system as a whole are quite intricate, despite the simplicity of the couplings. They depend, indeed, on the possibility of  tuning their  locality and strength and can be understood through their periods, Lyapunov spectra and Kolmogorov-Sinai entropies. An interesting point, which will be the subject of future work is the 
discrete conformal symmetry, that emerges, when the range of the couplings  becomes maximal.

A further topic is the construction of quantum Arnol'd cat  map lattices,  along with their continuum (scaling) limit as field theories.  It is here that the issues of closed subsystem thermalization dynamics, the Eigenstate Thermalization Hypothesis, as well as the saturation of the fast scrambling bound,  for many-body systems, can be framed and consistently treated. 
A useful diagnostic for this is  the set of  so-called ``out of time-order correlation functions''~\cite{Kurchan,Foini:2019tii}.

{\bf Acknowledgements:} This research was funded by the CNRS IEA (International Emerging Actions) program ``Chaotic behavior of closed quantum systems'' under contract 318687.
\newpage
\appendix
\section{The Fibonacci polynomials and their basic properties}\label{appkfibseq}
The Fibonacci sequence is one of the integer sequences, which has been studied, for a long time and there are journals dedicated to its properties and their applications. 
One generalization is provided by the sequence of polynomials, $f_m(x),$ defined by 
\begin{equation}
\label{fibinacci_seq}
\begin{array}{l}
f_0(x)=0; f_1(x)=1\\
f_{m+1}(x)= x f_m(x)+f_{m-1}(x)\\
\end{array}
\end{equation}
which can be written  in matrix form
\begin{equation}
\label{matrixfib}
\left(\begin{array}{c} f_m(x)\\f_{m+1}(x)\end{array}\right)=\underbrace{\left(\begin{array}{cc} 0 & 1 \\ 1 & x\end{array}\right)}_{{\sf A}(x)}\left(\begin{array}{c} f_{m-1}(x) \\ f_m(x)\end{array}\right)
\end{equation}
The matrix ${\sf A}(x)$ is not a symplectic matrix, but it satisfies 
\begin{equation}
\label{nonsymplA}
{\sf A}(x)^\mathrm{T}{\sf J}{\sf A}(x)=-{\sf J}
\end{equation}
where ${\sf J}$ is the symplectic matrix~(\ref{Jform}), for $n=1.$

The integer Fibonacci sequence, $f_m=f_m(x=1)$ and the integer $k-$Fibonacci sequence corresponds to $f_m(x=k)$ with $k=2,3,\ldots$

These sequences are related to  the ``golden'' and ``silver ratios'' by 
\begin{equation}
\label{fibratios}
\lim_{m\to\infty}\,\frac{f_{m+1}(x)}{f_m(x)}=\gamma(x)=\frac{x+\sqrt{x^2+4}}{2}
\end{equation}
for $x=1$ and 2 respectively. 

The Fibonacci polynomials are given explicitly by the relation
\begin{equation}
\label{mfibopoly}
f_m(x)=
\sum _{j=0}^{\left[\frac{m-1}{2}\right]} \binom{-j+m-1}{j} x^{-2 j+m-1},
\end{equation}
where $[\cdot]$ denotes the integer part of the argument. 

The Fibonacci polynomials are generated by powers of the matrix ${\sf A}(x)$ {\em viz.}
\begin{equation}
\label{FibpolyAofx}
{\sf A}(x)^m = \left(\begin{array}{cc}  f_{m-1}(x) & f_m(x)\\ f_m(x) & f_{m+1}(x)\end{array}\right)
\end{equation}
This relation is the origin, in fact, of the properties of the Fibonacci polynomials:
\begin{equation}
\label{propfibpoly}
\begin{array}{l}
\displaystyle
\mathrm{det}\,{\sf A}^m(x)=f_{m-1}(x)f_{m+1}(x)-f_m(x)^2=(-)^m\\
\displaystyle
{\sf A}^{pq}(x)=[{\sf A}^p(x)]^q=[{\sf A}^q(x)]^p\\

\displaystyle
{\sf A}^{p}(x) {\sf A}^{q}(x)={\sf A}^{p+q}(x)
\end{array}
\end{equation}
These imply that 
\begin{equation}
\label{propfibpoly1}
\left(\begin{array}{cc}  f_{pq-1}(x) & f_{pq}(x)\\ f_{pq}(x) & f_{pq+1}(x)\end{array}\right)= 
\left(\begin{array}{cc}  f_{p-1}(x) & f_p(x)\\ f_p(x) & f_{p+1}(x)\end{array}\right)^q=
\left(\begin{array}{cc}  f_{q-1}(x) & f_q(x)\\ f_q(x) & f_{q+1}(x)\end{array}\right)^p
\end{equation}

For future reference, we call ${\sf A}^2  ,$ ACM$_1$, to indicate that it describes the motion of a single particle. Below we shall study the dynamics of $n$ particles.

Since the matrix ${\sf A}(x)$ doesn't depend on $m,$ we can solve the recursion relation in closed form, by setting $f_m\equiv C \rho(x)^m$ and find the equation, satisfied by $\rho(x)$
$$
\rho(x)^{m+1}=x\rho(x)^m+\rho(x)^{m-1}\Leftrightarrow \rho(x)^2-x\rho(x)-1=0\Leftrightarrow \rho(x)\equiv \rho(x)_\pm=\frac{x\pm\sqrt{x^2+4}}{2}
$$
Therefore, we may express $f_m(x)$ as a linear combination of $\rho_+^m(x)$ and $\rho_-^m(x)=(-)^m\rho_+^{-m}(x)$:
\be
\label{solfib}
f_m(x)=A_+\rho_+^m(x)+A_-\rho_-^m(x)\Leftrightarrow\left\{\begin{array}{l} f_0=A_+ + A_- = 0\\ f_1 = A_+\rho(x)_+ + A_-\rho(x)_-=1\end{array}\right.
\ee
whence we find that 
$$
A_+ = -A_-=\frac{1}{\rho_+(x) -\rho_-(x)}=\frac{1}{\sqrt{x^2+4}}
$$
therefore,
\be
\label{solfib1}
f_m(x)=\frac{\rho(x)_+^m-(-)^m\rho(x)_+^{-m}}{\sqrt{x^2+4}}
\ee
It's quite fascinating that the LHS of this expression is a polynomial in $x$ that, moreover, takes integer values for integer values of $x$!

\newpage
\section{The conservation laws of ACML systems}\label{conslawsmodN}
In this section we study the conservation laws of the discrete evolution equations, 
\begin{equation}
\label{discreteevoleq}
\bm{q}_{m+1}-2\bm{q}_m+\bm{q}_{m-1}=\bm{q}_m{\sf C}^2
\end{equation}
where ${\sf C}$ is the symmetric, integer-valued matrix, given by the expression~(\ref{Ctransinv}), i.e.
\begin{equation}
\label{Ctransinv1} 
{\sf C}= KI_{n\times n} + \sum_{l=1}^{[n/2]-1}\,G_l\left( {\sf P}^l+\left[{\sf P}^\mathrm{T}\right]^l\right)
\end{equation}
In this equation $\bm{q}\in\mathbb{T}^n,$ whose compactness ensures mixing and implies that in eq.~(\ref{discreteevoleq}) there is an implicit ``mod 1'', to ensure that $\bm{q}_m\in\mathbb{T}^n$ for all timesteps.

A conservation law is related to, either, to the existence of an eigenvalue equal to 1 of the evolution operator, ${\sf M},$ or to a degeneracy of eigenvalues of ${\sf M}.$
These, in turn, can be recast in terms of properties of the matrix ${\sf C}.$The first case corresponds to a zero eigenvalue of ${\sf C}.$ The corresponding eigenvector, $\bm{a}\in\mathbb{T}^n,$ allows us to identify the symmetry as the translation $\bm{q}_m\to \bm{q}_m+\bm{a}.$ In this case, there exists a solution of Newton's equations, which is 
linear in time, {\em viz.}
\begin{equation}
\label{lineartimesol}
\bm{q}_m=m\bm{a}
\end{equation}
This means that the corresponding Lyapunov exponent  vanishes. 

For the matrix ${\sf C}$ defined by eq.~(\ref{Ctransinv1}), when $G_l=G$ for $l=1$ and $G_l=0$ for $l>1,$ (i.e. nearest-neighbor interactions) we can verify that  such an eigenvector always exists, when $K=2G,$ and  $n$ is even and is given by the expression 
\begin{equation}
\label{KGzeromodenn}
a_I = \frac{1}{\sqrt{n}}(-)^I
\end{equation}
where $I=1,2,\ldots,n.$

When couplings beyond those between nearest-neighbors are non-zero, $n$ must, still,  be even and the condition on the couplings, for the existence of a zeromode,  becomes
\begin{equation}
\label{KGzeomode}
K+2\sum_{l=1}^L G_l(-)^l=0
\end{equation}
Let us now consider the case, when the spectrum of the matrix ${\sf C}$ shows degeneracies, which, in turn, correspond to degeneracies in the spectrum of the Lyapunov exponents. 

We observe that translation invariance in the target space, $\bm{q}\to\bm{q}+\bm{a},$ corresponds to a symmetry that is an inhomogeneous transformation, which is the hallmark of zeromodes, while degeneracies correspond to symmetries given by homogeneous transformations. 

To look for such transformations we work as follows:

Since the target space is a torus, there exists a group of transformations, beyond the translations mod 1, namely the orthogonal group over the integers,  O$_n[\mathbb{Z}].$
By definition, an element ${\sf R}\in\mathrm{O}_n[\mathbb{Z}]$ satisfies the condition ${\sf R}{\sf R}^\mathrm{T}=I_{n\times n}.$

Applying such a transformation to the equation of motion, we find that the condition
\begin{equation}
\label{rotations}
{\sf R}^\mathrm{T}{\sf C}{\sf R}={\sf C}
\end{equation}
guarantees that these rotations are symmetries of the equations of motion. 

All such transformations define a subgroup of O$_n[\mathbb{Z}],$ that is, therefore the invariance group of the ACML. In this case, the existence of a zeromode $\bm{a}$ leads  to the existence of additional zeromodes, given by $\bm{a}{\sf R}.$ This defines a linear subspace of zeromodes, labeled by all such matrices ${\sf R}.$ 

That this group isn't empty follows from the particular form of the matrix  ${\sf C},$ which commutes with the matrix ${\sf P}.$  The matrix ${\sf P}$ belongs to O$_n[\mathbb{Z}]$ and represents, by a shift along the lattice, a rotation in the target space! 

This particular symmetry is the reason for the degeneracy in the spectrum, ${\sf D}_I, I=1,2,\ldots,n$ namely
\begin{equation}
\label{Ddegeneracy}
D_I = D_{n-I}
\end{equation}
So, finally, degeneracies of the spectrum of  ${\sf C},$ are explained by the existence of rotations ${\sf R}\in\mathrm{O}_n[\mathbb{Z}],$ which commute with ${\sf P}.$

The above transformations describe discrete  spatial translations as well as rotations in the target space; however, there exists a further, important, symmetry, of the equations of motion~(\ref{discreteevoleq}), namely that of discrete translations in time, $m\to m+1.$

An immediate consequence of this symmetry is that, from the block form for ${\sf M},$
\begin{equation}
\label{blockM}
{\sf M}=\left(\begin{array}{cc} {\sf A} & {\sf B}\\ {\sf C} & {\sf D}\end{array}\right)
\end{equation}
we find that the quadratic form, $Q(\bm{x}),$  given by the expression
\begin{equation}
\label{conservedQofx}
Q(\bm{x})=\bm{q}^\mathrm{T}{\sf B}^\mathrm{T}\bm{q}+\bm{q}^\mathrm{T}\left({\sf D}^\mathrm{T}-{\sf A}\right)\bm{p}-\bm{p}^\mathrm{T}{\sf C}^\mathrm{T}\bm{p}
\end{equation}
is conserved in time:
\begin{equation}
\label{Qmplus1=Qm}
Q(\bm{x}_{m+1})=Q(\bm{x}_m)
\end{equation}
where $\bm{x}_m=(\bm{q}_m,\bm{p}_m).$

For the particular case of the CACML, ${\sf A}=I_{n\times n},$ ${\sf B}={\sf C}$ and ${\sf D}=1+{\sf C}^2,$ we find the expression 
\begin{equation}
\label{conservedQACML}
Q(\bm{x})=\bm{q}^\mathrm{T}{\sf C}\bm{q}+\bm{q}^\mathrm{T}{\sf C}^2\bm{p}-\bm{p}^\mathrm{T}{\sf C}\bm{p}
\end{equation}
since ${\sf C}$ is symmetric.

Of course what is significant is that the global sign is arbitrary, therefore, we can write it in a more ``conventional'' form as
\begin{equation}
\label{conservedQACML1}
Q(\bm{x})=\bm{p}^\mathrm{T}{\sf C}\bm{p}-\bm{q}^\mathrm{T}{\sf C}\bm{q}-\bm{q}^\mathrm{T}{\sf C}^2\bm{p}
\end{equation}
where the mod 1 operation is implicit. 

This, particular, solution will play a significant role in the construction of the quantum, unitary, evolution operator $U({\sf M})$ for the corresponding quantum mechanical system of $n$ coupled Arnol'd cat maps.

\newpage
\bibliographystyle{utphys}
\bibliography{ads2discrete}

\providecommand{\href}[2]{#2}\begingroup\raggedright\begin{thebibliography}{10}

\bibitem{arnold2007mathematical}
V.~I. Arnold, V.~V. Kozlov, and A.~I. Neishtadt, {\em Mathematical aspects of
  classical and celestial mechanics}, vol.~3.
\newblock Springer Science \& Business Media, 2007.

\bibitem{cvitanovic2005chaos}
P.~Cvitanovic, R.~Artuso, R.~Mainieri, G.~Tanner, G.~Vattay, N.~Whelan, and
  A.~Wirzba, {\em Chaos: classical and quantum}.
\newblock 2005.
\newblock \url{https://chaosbook.org/}.

\bibitem{zaslavsky2008}
G.~Zaslavsky, {\em Hamiltonian chaos and fractional dynamics}.
\newblock Oxford University Press, 2008.

\bibitem{gutzwiller2013chaos}
M.~C. Gutzwiller, {\em Chaos in classical and quantum mechanics}, vol.~1.
\newblock Springer Science \& Business Media, 2013.

\bibitem{crutchfield1987phenomenology}
J.~P. Crutchfield and K.~Kaneko, {\em Phenomenology of spatio-temporal chaos}.
\newblock World Scientific, 1987.

\bibitem{kaneko2001complex}
K.~Kaneko and I.~Tsuda, {\em Complex Systems: Chaos and Beyond: Chaos and
  Beyond: A Constructive Approach With Applications in Life Sciences}.
\newblock Springer Science \& Business Media, 2001.

\bibitem{vulpiani2009chaos}
A.~Vulpiani, F.~Cecconi, and M.~Cencini, {\em Chaos: from simple models to
  complex systems}, vol.~17.
\newblock World Scientific, 2009.

\bibitem{badiicomplexity}
R.~Badii and A.~Politi, {\em Complexity hierarchical structures and scaling in
  physics, 1997}.
\newblock Cambridge University Press, Cambridge, UK.

\bibitem{cross1994spatiotemporal}
M.~Cross and P.~Hohenberg, ``Spatiotemporal chaos,'' {\em Science-AAAS-Weekly
  Paper Edition-including Guide to Scientific Information} {\bfseries 263}
  no.~5153, (1994) 1569--1569.

\bibitem{cross1993pattern}
M.~C. Cross and P.~C. Hohenberg, ``Pattern formation outside of equilibrium,''
  \href{http://dx.doi.org/10.1103/RevModPhys.65.851}{{\em Rev. Mod. Phys.}
  {\bfseries 65} (Jul, 1993) 851--1112}.
  \url{https://link.aps.org/doi/10.1103/RevModPhys.65.851}.

\bibitem{Kaneko:2014}
K.~Kaneko and T.~Yanagita, ``{C}oupled maps,''
  \href{http://dx.doi.org/10.4249/scholarpedia.4085}{{\em Scholarpedia}
  {\bfseries 9} no.~5, (2014) 4085}. revision \#149460.

\bibitem{pikovsky2016lyapunov}
A.~Pikovsky and A.~Politi, {\em Lyapunov exponents: a tool to explore complex
  dynamics}.
\newblock Cambridge University Press, 2016.

\bibitem{Sinai1994}
Y.~G. Sinai, {\em Topics in Ergodic Theory, {\rm Princeton Mathematical Series,
  vol. 44 (PMS-44)}}.
\newblock Princeton University Press, 1994.
\newblock \url{http://www.jstor.org/stable/j.ctt1m321xp}.

\bibitem{cornfeld2012ergodic}
I.~P. Cornfeld, S.~V. Fomin, and Y.~G. Sinai, {\em Ergodic theory}, vol.~245.
\newblock Springer Science \& Business Media, 2012.

\bibitem{cvitanovic2000chaotic}
P.~Cvitanovi{\'c}, ``Chaotic field theory: a sketch,'' {\em Physica A:
  Statistical Mechanics and its Applications} {\bfseries 288} no.~1-4, (2000)
  61--80.

\bibitem{gutkin2016}
B.~Gutkin and V.~Osipov, ``Classical foundations of many-particle quantum
  chaos,'' \href{http://dx.doi.org/10.1088/0951-7715/29/2/325}{{\em
  Nonlinearity} {\bfseries 29} no.~2, (Jan, 2016) 325–356}.
  \url{http://dx.doi.org/10.1088/0951-7715/29/2/325}.

\bibitem{gutkin:2019sca}
B.~Gutkin, P.~Cvitanovi\'c, R.~Jafari, A.~K. Saremi, and L.~Han, ``{Linear
  encoding of the spatiotemporal cat},''
  \href{http://dx.doi.org/10.1088/1361-6544/abd7c8}{{\em Nonlinearity}
  {\bfseries 34} no.~5, (2021) 2800--2836},
  \href{http://arxiv.org/abs/1912.02940}{{\ttfamily arXiv:1912.02940
  [nlin.CD]}}.

\bibitem{cvitanovic2020spatiotemporal}
P.~Cvitanovic and H.~Liang, ``Spatiotemporal cat: a chaotic field theory,''.
  \url{https://chaosbook.org/overheads/spatiotemporal/CL18.pdf}.

\bibitem{holmes2012turbulence}
P.~Holmes, J.~L. Lumley, G.~Berkooz, and C.~W. Rowley, {\em Turbulence,
  coherent structures, dynamical systems and symmetry}.
\newblock {Cambridge University Press}, 2012.

\bibitem{auerbach1987exploring}
D.~Auerbach, P.~Cvitanovi{\'c}, J.-P. Eckmann, G.~Gunaratne, and I.~Procaccia,
  ``Exploring chaotic motion through periodic orbits,'' {\em Physical Review
  Letters} {\bfseries 58} no.~23, (1987) 2387.

\bibitem{Floratos:2005yj}
E.~G. Floratos and S.~Nicolis, ``{Unitary evolution on a discrete phase
  space},'' \href{http://dx.doi.org/10.22323/1.020.0260}{{\em PoS} {\bfseries
  LAT2005} (2006) 260}, \href{http://arxiv.org/abs/hep-lat/0510043}{{\ttfamily
  arXiv:hep-lat/0510043}}.

\bibitem{kaiblinger2009metaplectic}
N.~Kaiblinger and M.~Neuhauser, ``Metaplectic operators for finite abelian
  groups and $\mathbb{R}^d$,'' {\em Indagationes Mathematicae} {\bfseries 20}
  no.~2, (2009) 233--246.

\bibitem{sep-ergodic-hierarchy}
R.~Frigg, J.~Berkovitz, and F.~Kronz, ``{The Ergodic Hierarchy},'' in {\em The
  {Stanford} Encyclopedia of Philosophy}, E.~N. Zalta, ed.
\newblock Metaphysics Research Lab, Stanford University, {F}all 2020~ed., 2020.

\bibitem{Axenides:2019lea}
M.~Axenides, E.~Floratos, and S.~Nicolis, ``{The arithmetic geometry of AdS$_2$
  and its continuum limit},''
  \href{http://dx.doi.org/10.3842/SIGMA.2021.004}{{\em Symmetry, Integrability
  and Geometry: Methods and Applications (SIGMA)} {\bfseries 17} (2021) 004},
  \href{http://arxiv.org/abs/1908.06641}{{\ttfamily arXiv:1908.06641
  [hep-th]}}.

\bibitem{percival1987linear}
I.~Percival and F.~Vivaldi, ``A linear code for the sawtooth and cat maps,''
  {\em Physica D: Nonlinear Phenomena} {\bfseries 27} no.~3, (1987) 373--386.

\bibitem{percival1987arithmetical}
I.~Percival and F.~Vivaldi, ``Arithmetical properties of strongly chaotic
  motions,'' {\em Physica D: Nonlinear Phenomena} {\bfseries 25} no.~1-3,
  (1987) 105--130.

\bibitem{bird1988periodic}
N.~Bird and F.~Vivaldi, ``Periodic orbits of the sawtooth maps,'' {\em Physica
  D: Nonlinear Phenomena} {\bfseries 30} no.~1-2, (1988) 164--176.

\bibitem{Keating1991}
J.~P. Keating, ``Asymptotic properties of the periodic orbits of the cat
  maps,'' \href{http://dx.doi.org/10.1088/0951-7715/4/2/005}{{\em Nonlinearity}
  {\bfseries 4} no.~2, (May, 1991) 277--307}.
  \url{https://doi.org/10.1088/0951-7715/4/2/005}.

\bibitem{Hayden:2007cs}
P.~Hayden and J.~Preskill, ``{Black holes as mirrors: Quantum information in
  random subsystems},''
  \href{http://dx.doi.org/10.1088/1126-6708/2007/09/120}{{\em JHEP} {\bfseries
  09} (2007) 120},
\href{http://arxiv.org/abs/0708.4025}{{\ttfamily arXiv:0708.4025 [hep-th]}}.

\bibitem{Sekino:2008he}
Y.~Sekino and L.~Susskind, ``{Fast Scramblers},''
  \href{http://dx.doi.org/10.1088/1126-6708/2008/10/065}{{\em JHEP} {\bfseries
  10} (2008) 065},
\href{http://arxiv.org/abs/0808.2096}{{\ttfamily arXiv:0808.2096 [hep-th]}}.

\bibitem{Shenker:2013pqa}
S.~H. Shenker and D.~Stanford, ``{Black holes and the butterfly effect},''
  \href{http://dx.doi.org/10.1007/JHEP03(2014)067}{{\em JHEP} {\bfseries 03}
  (2014) 067},
\href{http://arxiv.org/abs/1306.0622}{{\ttfamily arXiv:1306.0622 [hep-th]}}.

\bibitem{Maldacena:2015waa}
J.~Maldacena, S.~H. Shenker, and D.~Stanford, ``{A bound on chaos},''
  \href{http://dx.doi.org/10.1007/JHEP08(2016)106}{{\em JHEP} {\bfseries 08}
  (2016) 106},
\href{http://arxiv.org/abs/1503.01409}{{\ttfamily arXiv:1503.01409 [hep-th]}}.

\bibitem{Bousso:2022ntt}
R.~Bousso, X.~Dong, N.~Engelhardt, T.~Faulkner, T.~Hartman, S.~H. Shenker, and
  D.~Stanford, ``{Snowmass White Paper: Quantum Aspects of Black Holes and the
  Emergence of Spacetime},'' \href{http://arxiv.org/abs/2201.03096}{{\ttfamily
  arXiv:2201.03096 [hep-th]}}.

\bibitem{Hooft:2015jea}
G.~'t~Hooft, ``{Diagonalizing the Black Hole Information Retrieval Process},''
\href{http://arxiv.org/abs/1509.01695}{{\ttfamily arXiv:1509.01695 [gr-qc]}}.

\bibitem{Barrabes:2000fr}
C.~Barrabes, V.~P. Frolov, and R.~Parentani, ``{Stochastically fluctuating
  black hole geometry, Hawking radiation and the transPlanckian problem},''
  \href{http://dx.doi.org/10.1103/PhysRevD.62.044020}{{\em Phys. Rev.}
  {\bfseries D62} (2000) 044020},
\href{http://arxiv.org/abs/gr-qc/0001102}{{\ttfamily arXiv:gr-qc/0001102
  [gr-qc]}}.

\bibitem{Papadodimas:2012aq}
K.~Papadodimas and S.~Raju, ``{An Infalling Observer in AdS/CFT},''
  \href{http://dx.doi.org/10.1007/JHEP10(2013)212}{{\em JHEP} {\bfseries 10}
  (2013) 212},
\href{http://arxiv.org/abs/1211.6767}{{\ttfamily arXiv:1211.6767 [hep-th]}}.

\bibitem{Banerjee:2016mhh}
S.~Banerjee, J.-W. Bryan, K.~Papadodimas, and S.~Raju, ``{A toy model of black
  hole complementarity},''
  \href{http://dx.doi.org/10.1007/JHEP05(2016)004}{{\em JHEP} {\bfseries 05}
  (2016) 004},
\href{http://arxiv.org/abs/1603.02812}{{\ttfamily arXiv:1603.02812 [hep-th]}}.

\bibitem{Stanford:2014jda}
D.~Stanford and L.~Susskind, ``{Complexity and Shock Wave Geometries},''
  \href{http://dx.doi.org/10.1103/PhysRevD.90.126007}{{\em Phys. Rev.}
  {\bfseries D90} no.~12, (2014) 126007},
\href{http://arxiv.org/abs/1406.2678}{{\ttfamily arXiv:1406.2678 [hep-th]}}.

\bibitem{Shenker:2013yza}
S.~H. Shenker and D.~Stanford, ``{Multiple Shocks},''
  \href{http://dx.doi.org/10.1007/JHEP12(2014)046}{{\em JHEP} {\bfseries 12}
  (2014) 046},
\href{http://arxiv.org/abs/1312.3296}{{\ttfamily arXiv:1312.3296 [hep-th]}}.

\bibitem{Mezei:2016wfz}
M.~Mezei and D.~Stanford, ``{On entanglement spreading in chaotic systems},''
  \href{http://dx.doi.org/10.1007/JHEP05(2017)065}{{\em JHEP} {\bfseries 05}
  (2017) 065},
\href{http://arxiv.org/abs/1608.05101}{{\ttfamily arXiv:1608.05101 [hep-th]}}.

\bibitem{Polchinski:2015cea}
J.~Polchinski, ``{Chaos in the black hole S-matrix},''
\href{http://arxiv.org/abs/1505.08108}{{\ttfamily arXiv:1505.08108 [hep-th]}}.

\bibitem{Strominger:2014pwa}
A.~Strominger and A.~Zhiboedov, ``{Gravitational Memory, BMS Supertranslations
  and Soft Theorems},'' \href{http://dx.doi.org/10.1007/JHEP01(2016)086}{{\em
  JHEP} {\bfseries 01} (2016) 086},
\href{http://arxiv.org/abs/1411.5745}{{\ttfamily arXiv:1411.5745 [hep-th]}}.

\bibitem{Ellis:2016atb}
J.~Ellis, N.~E. Mavromatos, and D.~V. Nanopoulos, ``{$W_\infty$ Algebras,
  Hawking Radiation and Information Retention by Stringy Black Holes},''
  \href{http://dx.doi.org/10.1103/PhysRevD.94.025007}{{\em Phys. Rev.}
  {\bfseries D94} no.~2, (2016) 025007},
\href{http://arxiv.org/abs/1605.01653}{{\ttfamily arXiv:1605.01653 [hep-th]}}.

\bibitem{Hawking:2016msc}
S.~W. Hawking, M.~J. Perry, and A.~Strominger, ``{Soft Hair on Black Holes},''
  \href{http://dx.doi.org/10.1103/PhysRevLett.116.231301}{{\em Phys. Rev.
  Lett.} {\bfseries 116} no.~23, (2016) 231301},
\href{http://arxiv.org/abs/1601.00921}{{\ttfamily arXiv:1601.00921 [hep-th]}}.

\bibitem{Hooft:2016pmw}
G.~'t~Hooft, ``{How quantization of gravity leads to a discrete space-time},''
\href{http://dx.doi.org/10.1088/1742-6596/701/1/012014}{{\em J. Phys. Conf.
  Ser.} {\bfseries 701} no.~1, (2016) 012014}.

\bibitem{Brown:2016wib}
A.~R. Brown, L.~Susskind, and Y.~Zhao, ``{Quantum Complexity and Negative
  Curvature},'' \href{http://dx.doi.org/10.1103/PhysRevD.95.045010}{{\em Phys.
  Rev.} {\bfseries D95} no.~4, (2017) 045010},
\href{http://arxiv.org/abs/1608.02612}{{\ttfamily arXiv:1608.02612 [hep-th]}}.

\bibitem{Banks:1997hz}
T.~Banks, W.~Fischler, I.~R. Klebanov, and L.~Susskind, ``{Schwarzschild black
  holes from matrix theory},''
  \href{http://dx.doi.org/10.1103/PhysRevLett.80.226}{{\em Phys. Rev. Lett.}
  {\bfseries 80} (1998) 226--229},
\href{http://arxiv.org/abs/hep-th/9709091}{{\ttfamily arXiv:hep-th/9709091
  [hep-th]}}.

\bibitem{Iizuka:2008eb}
N.~Iizuka, T.~Okuda, and J.~Polchinski, ``{Matrix Models for the Black Hole
  Information Paradox},'' \href{http://dx.doi.org/10.1007/JHEP02(2010)073}{{\em
  JHEP} {\bfseries 02} (2010) 073},
\href{http://arxiv.org/abs/0808.0530}{{\ttfamily arXiv:0808.0530 [hep-th]}}.

\bibitem{Avery:2011nb}
S.~G. Avery, ``{Qubit Models of Black Hole Evaporation},''
  \href{http://dx.doi.org/10.1007/JHEP01(2013)176}{{\em JHEP} {\bfseries 01}
  (2013) 176},
\href{http://arxiv.org/abs/1109.2911}{{\ttfamily arXiv:1109.2911 [hep-th]}}.

\bibitem{Magan:2016ojb}
J.~M. Magan, ``{Black holes as random particles: entanglement dynamics in
  infinite range and matrix models},''
  \href{http://dx.doi.org/10.1007/JHEP08(2016)081}{{\em JHEP} {\bfseries 08}
  (2016) 081},
\href{http://arxiv.org/abs/1601.04663}{{\ttfamily arXiv:1601.04663 [hep-th]}}.

\bibitem{Axenides:2013iwa}
M.~Axenides, E.~G. Floratos, and S.~Nicolis, ``{Modular discretization of the
  AdS$_{2}$/CFT$_{1}$ holography},''
  \href{http://dx.doi.org/10.1007/JHEP02(2014)109}{{\em JHEP} {\bfseries 02}
  (2014) 109},
\href{http://arxiv.org/abs/1306.5670}{{\ttfamily arXiv:1306.5670 [hep-th]}}.

\bibitem{Axenides:2015aha}
M.~Axenides, E.~Floratos, and S.~Nicolis, ``{Chaotic Information Processing by
  Extremal Black Holes},''
  \href{http://dx.doi.org/10.1142/S0218271815420122}{{\em Int. J. Mod. Phys.}
  {\bfseries D24} no.~09, (2015) 1542012},
\href{http://arxiv.org/abs/1504.00483}{{\ttfamily arXiv:1504.00483 [hep-th]}}.

\bibitem{Axenides:2016nmf}
M.~Axenides, E.~Floratos, and S.~Nicolis, ``{The quantum cat map on the modular
  discretization of extremal black hole horizons},''
  \href{http://dx.doi.org/10.1140/epjc/s10052-018-5850-9}{{\em Eur. Phys. J.}
  {\bfseries C78} no.~5, (2018) 412},
\href{http://arxiv.org/abs/1608.07845}{{\ttfamily arXiv:1608.07845 [hep-th]}}.

\bibitem{SrednickiETH}
M.~Srednicki, ``Chaos and quantum thermalization,''
  \href{http://dx.doi.org/10.1103/PhysRevE.50.888}{{\em Phys. Rev. E}
  {\bfseries 50} (Aug, 1994) 888--901}.
  \url{https://link.aps.org/doi/10.1103/PhysRevE.50.888}.

\bibitem{Arnold11}
V.~I. Arnol'd, {\em {Dynamics, statistics and projective geometry of Galois
  fields}}.
\newblock Cambridge University Press, 2011.

\bibitem{mantica2019many}
G.~Mantica, ``{Many-Body Systems and Quantum Chaos: The Multiparticle Quantum
  Arnol’d Cat},'' {\em Condensed Matter} {\bfseries 4} no.~3, (2019) 72.

\bibitem{mcduff2017introduction}
D.~McDuff and D.~Salamon, {\em Introduction to symplectic topology}, vol.~27.
\newblock Oxford University Press, 2017.

\bibitem{Bianchi:2017kgb}
E.~Bianchi, L.~Hackl, and N.~Yokomizo, ``{Linear growth of the entanglement
  entropy and the Kolmogorov-Sinai rate},''
  \href{http://dx.doi.org/10.1007/JHEP03(2018)025}{{\em JHEP} {\bfseries 03}
  (2018) 025},
\href{http://arxiv.org/abs/1709.00427}{{\ttfamily arXiv:1709.00427 [hep-th]}}.

\bibitem{de2022linear}
G.~De~Palma and L.~Hackl, ``{Linear growth of the entanglement entropy for
  quadratic Hamiltonians and arbitrary initial states},'' {\em SciPost Physics}
  {\bfseries 12} no.~1, (2022) 021.

\bibitem{Athanasiu:1995ni}
G.~G. Athanasiu, E.~G. Floratos, and S.~Nicolis, ``{Holomorphic quantization on
  the torus and finite quantum mechanics},''
  \href{http://dx.doi.org/10.1088/0305-4470/29/21/010}{{\em J. Phys. A}
  {\bfseries 29} (1996) 6737},
  \href{http://arxiv.org/abs/hep-th/9509098}{{\ttfamily arXiv:hep-th/9509098}}.

\bibitem{falk_dyson}
F.~J. Dyson and H.~Falk, ``Period of a discrete cat mapping,'' {\em The
  American Mathematical Monthly} {\bfseries 99} no.~7, (1992) 603--614.
  \url{http://www.jstor.org/stable/2324989}.

\bibitem{Athanasiu:1998cq}
G.~G. Athanasiu, E.~G. Floratos, and S.~Nicolis, ``{Fast quantum maps},''
  \href{http://dx.doi.org/10.1088/0305-4470/31/38/001}{{\em J. Phys.}
  {\bfseries A31} (1998) L655},
\href{http://arxiv.org/abs/math-ph/9805012}{{\ttfamily arXiv:math-ph/9805012
  [math-ph]}}.

\bibitem{sp2norder}
``{Order formulas for symplectic groups}.''
  \url{https://groupprops.subwiki.org/wiki/Order_formulas_for_symplectic_groups}.

\bibitem{atanassov1985new}
K.~T. Atanassov, L.~Atanassov, and D.~D. Sasselov, ``{A new perspective to the
  generalization of the Fibonacci sequence},'' {\em The Fibonacci Quarterly}
  {\bfseries 23} no.~1, (1985) 21--28.

\bibitem{singh2010coupled}
M.~Singh, O.~Sikhwal, and S.~Jain, ``{Coupled Fibonacci sequence of fifth order
  and some properties},'' {\em International Journal of Mathematical Analysis}
  {\bfseries 4} no.~25, (2010) 1247--1254.

\bibitem{rathore2012generalized}
G.~Rathore, S.~Jain, and O.~Sikhwal, ``{Generalized Coupled Fibonacci
  Sequences},'' {\em International Journal of Computer Applications} {\bfseries
  59} no.~8, (2012) .

\bibitem{pesin1977characteristic}
Y.~B. Pesin, ``{Characteristic Lyapunov exponents and smooth ergodic theory},''
  {\em Russian Mathematical Surveys} {\bfseries 32} no.~4, (1977) 55.

\bibitem{sinai1989dynamical}
I.~G. Sinai, {\em Dynamical systems II: Ergodic theory with applications to
  dynamical systems and statistical mechanics}.
\newblock Springer, 1989.

\bibitem{AFNmixing}
M.~Axenides, E.~Floratos, and S.~Nicolis, ``{The mixing time for Arnol'd cat
  map lattices},'' {\em work in progress} .

\bibitem{arnold1968ergodic}
V.~I. Arnold and A.~Avez, {\em Ergodic problems of classical mechanics},
  vol.~9.
\newblock Benjamin, 1968.

\bibitem{crawford1983decay}
J.~D. Crawford and J.~R. Cary, ``Decay of correlations in a chaotic
  measure-preserving transformation,'' {\em Physica D: Nonlinear Phenomena}
  {\bfseries 6} no.~2, (1983) 223--232.

\bibitem{latora1999kolmogorov}
V.~Latora and M.~Baranger, ``{Kolmogorov-Sinai entropy rate versus physical
  entropy},'' {\em Physical Review Letters} {\bfseries 82} no.~3, (1999) 520.

\bibitem{philippou2001fibonacci}
A.~Philippou, ``Fibonacci polynomials,'' {\em Invited contribution in
  Encyclopedia of Mathematics, Supplement III, Kluwer Academic Publishers, {\rm
  ed. M. Hazewinkel}} (2001) 152--154.

\bibitem{Kurchan}
L.~Foini and J.~Kurchan, ``Eigenstate thermalization hypothesis and out of time
  order correlators,'' \href{http://dx.doi.org/10.1103/PhysRevE.99.042139}{{\em
  Phys. Rev. E} {\bfseries 99} (Apr, 2019) 042139}.
  \url{https://link.aps.org/doi/10.1103/PhysRevE.99.042139}.

\bibitem{Foini:2019tii}
L.~Foini and J.~Kurchan, ``{Eigenstate Thermalization and Rotational Invariance
  in Ergodic Quantum Systems},''
  \href{http://dx.doi.org/10.1103/PhysRevLett.123.260601}{{\em Phys. Rev.
  Lett.} {\bfseries 123} no.~26, (2019) 260601},
  \href{http://arxiv.org/abs/1906.01522}{{\ttfamily arXiv:1906.01522
  [cond-mat.stat-mech]}}.

\end{thebibliography}\endgroup
\end{document}